\documentclass[smallextended,numbook,runningheads]{svjour3}     
\usepackage[sort,nocompress]{cite}
\smartqed  
\usepackage{graphicx}
\usepackage{amsmath}
\usepackage{epstopdf} 
\usepackage{mathptmx}

\usepackage{amsmath,amssymb,blkarray,color,enumerate,listings,mathtools,tabularx,url,xspace,array,amsfonts,xcolor}
\usepackage{multicol,multirow}        

\makeatletter
\def\rf#1{(\@rf#1,.)}
\def\@rf#1,{\ref{eq:#1}\@ifnextchar . {\@endrf}{, \@rf}}
\def\@endrf.{}
\makeatother

\usepackage{float}
\usepackage{hyperref}
\hypersetup{colorlinks=true,
  linkcolor=blue,
  citecolor=blue,
  filecolor=magenta,
  urlcolor=blue
  }

\usepackage{caption}
\usepackage[font=normalsize]{subcaption}

\newcommand\qaq{\quad\text{and}\quad}
\newcommand\dof{\text{{\sc dof}}\xspace}
\renewcommand{\~}[1]{\mathbf{#1}}
\newcommand{\setbg}[1]{\bigl\{\,#1\,\bigr\}}
\newcommand{\casemod}[2]{\left\{ \begin{array}{#1} #2 \end{array}\right. }
\newcommand{\pendL}{L}
\newcommand{\pendG}{g}

\newcommand{\Sigmeth}{$\Sigma$-method\xspace}
\newcommand{\mc}[3]{\multicolumn{#1}{#2}{#3}}

\newcommand{\s}[1]{\text{\scriptsize$#1$}}
\newcommand{\valSig}[2]{\val{#1}=#2}
\newcommand{\detJac}[2]{\det (#1)=#2}

\newcommand\offsetpair{(\~c;\~d)}

\newcommand{\sij}[2]{\sigma_{#1#2}}

\newcommand{\Sig}{\Sigma}
\newcommand{\Jac}{\~J}
\newcommand{\Jij}[2]{\Jac_{#1#2}}

\newcommand\ssth[1]{#1{\text{th}}}

\newcommand{\compSig}{\widetilde\Sigma}
\newcommand{\compJ}{\~{\widetilde J}}
\newcommand{\compJij}[2]{\~{\widetilde J}_{#1#2}}
\newcommand{\compsij}[2]{\widetilde\sigma_{#1#2}}

\newcommand\compoffvec{(\compc,\compd)}

\newcommand{\newsig}{\overline\sigma}
\newcommand{\newsij}[2]{\newsig_{#1#2}}

\newcommand{\newSig}{\overline\Sigma}
\newcommand{\newSigma}{\overline\Sigma}
\newcommand{\newJ}{\~{\overline J}}
\newcommand{\newJac}{\~{\overline J}}

\newcommand{\newT}{\overline T}
\newcommand{\newf}{\overline f}

\newcommand{\newc}{\overline c}
\newcommand{\newd}{\overline d}

\newcommand{\rnge}[2]{#1\,{:}\,#2}
\newcommand{\setrnge}[2]{\{#1,\ldots,#2\}}

\newcommand{\oneton}{\rnge{1}{n}}

\newcommand{\clsp}{\phantom{\ }}
\newcommand{\lam}{\lambda}

\newcommand{\daesa}{{\sc daesa}\xspace}

\newcommand{\daets}{{\sc daets}\xspace}

\newcommand{\indxk}{l}
\newcommand{\nzset}{L}
\newcommand{\eqsetI}{I}

\newcommand{\hoder}[2]{\sigma\left(#1,#2\right)}

\newcommand\eqsetIES{M}

\newcommand\nzsetES{J}

\newcommand{\jone}{j}
\newcommand{\jtwo}{r}
\newcommand\vecu{\~u}
\newcommand\vecv{\~v}

\newcommand{\LargeLess}{\parbox{12pt}{\Large $<$}}
\newcommand{\HugeLess}{\parbox{12pt}{\Large $<$}}
\newcommand{\HugeLe}{\parbox{12pt}{\Large $\le$}}
\newcommand{\HugeInf}{\parbox{12pt}{\Large $-\infty$}}

\newcommand{\val}[1]{\textnormal{Val}(#1)}
\newcommand{\set}[1]{\{\, #1 \,\}}
\newcommand{\pp}[2]{\frac{\partial #1}{\partial #2}}
\newcommand{\ppin}[2]{\partial #1/\partial #2}
\newcommand{\fracin}[2]{#1/#2}

\newcommand\nequiv{\not\equiv}

\definecolor{lgray}{rgb}{0.6,0.6,0.6}

\newcommand{\lgo}{\OK{0}}

\definecolor{palegray}{rgb}{0.7,0.7,0.7}
\newcommand{\OK}[1]{\colorbox{palegray}{\ensuremath{#1}}}

\newcommand{\rrf}[2]{(\ref{eq:#1}--\ref{eq:#2})}

\newcommand{\SSCref}[1]{\S\ref{ssc:#1}}

\newcommand{\SCref}[1]{\S\ref{sc:#1}}
\newcommand{\EXref}[1]{Example~\ref{ex:#1}}

\newcommand{\APref}[1]{\S\ref{ap:#1}}
\newcommand{\TBref}[1]{Table~\ref{tb:#1}}

\newcommand{\LEref}[1]{Lemma~\ref{le:#1}}
\newcommand{\FGref}[1]{Figure~\ref{fg:#1}}

\newcommand{\THref}[1]{Theorem~\ref{th:#1}}

\newcommand{\uiter}[1]{\~u}

\newcommand{\Iiter}[1]{\eqsetI}
\newcommand{\thiter}[1]{\LCconst}
\newcommand{\Siter}[1]{\Sig^#1}
\newcommand{\Jiter}[1]{\Jac^#1}

\newcommand{\Kiter}[1]{\nzset}

\newcommand{\kiter}[1]{\indxk}

\newcommand\as{\leftarrow}
\newcommand\phanbull{\phantom{^\bullet}}
\newcommand\LCconst{\underline{c}}
\newcommand\ESconst{\overline{c}}
\newcommand\underC{\LCconst}
\newcommand\overC{\ESconst}
\newcommand\newL{{\overline \nzset}}
\newcommand\newJes{{\overline \nzsetES}}

\newcommand\coker[1]{\textnormal{coker}(#1)}

\renewcommand\ker[1]{\textnormal{ker}(#1)}
\renewcommand\SCref[1]{Section~\ref{sc:#1}}
\renewcommand\newc{{\widetilde c}}
\renewcommand\newd{{\widetilde d}}
\newcommand\Jnol{\nzsetES\setminus\setbg\indxk}
\renewcommand{\Jij}[2]{J_{#1#2}}
\renewcommand{\compJij}[2]{\widetilde{J}_{#1#2}}

\renewcommand\compoffvec{(\widetilde{\~c};\widetilde{\~d})}
\newcommand\newSigBlk[2]{\newSig_{#1#2}}
\renewcommand{\APref}[1]{Appendix~\ref{ap:#1}}

\journalname{BIT}
\begin{document}

\title{Conversion Methods for Improving Structural Analysis of Differential-Algebraic Equation Systems
}

\titlerunning{Conversion Methods for Improving Structural Analysis of DAEs}        

\author{Guangning Tan         \and
Nedialko S. Nedialkov \and\\
John D. Pryce
}


\institute{G. Tan
\at School of Computational Science and Engineering, McMaster University, \at
1280 Main Street West, Hamilton, Ontario L8S 4K1, Canada,
\email{tang4@mcmaster.ca} \and
N. S. Nedialkov
\at Department of Computing and Software, McMaster University, \at
1280 Main Street West, , Hamilton, Ontario L8S 4K1, Canada,
\email{nedialk@mcmaster.ca} \and
J. D. Pryce
\at School of Mathematics, Cardiff University, \at Senghennydd Road, Cardiff CF24 4AG, Wales, UK.,
\email{prycejd1@cardiff.ac.uk}
}

\date{Received: date / Accepted: date}

\maketitle

\begin{abstract}
Differential-algebraic equation systems (DAEs) are generated routinely by simulation and modeling environments. Before a simulation starts and a numerical method is applied, some kind of structural analysis (SA) is used to determine which equations to be differentiated, and how many times. 
Both Pantelides's algorithm and Pryce's $\Sigma$-method are equivalent: if one of them finds correct structural information, the other does also. Nonsingularity of the Jacobian produced by SA indicates a success, which occurs on many problems of interest. However, these methods can fail on simple, solvable DAEs and give incorrect structural information including the index. This article investigates $\Sigma$-method's failures and presents two conversion methods for fixing them. Both methods convert a DAE on which the $\Sigma$-method fails to an equivalent problem on which this SA is more likely to succeed. 

\keywords{differential-algebraic equations \and structural analysis \and modeling\and symbolic computation}
\subclass{34A09\and 65L80\and 41A58 \and 68W30}
\end{abstract}


\section{Introduction.}\label{sc:intro}

Differential-algebraic equation systems (DAEs) arise from disciplines such as electrical circuits, chemical engineering, optimal control, and mechanical systems. 
To simulate the dynamic behaviour of such systems, a variety of algorithms are developed from building a mathematical model to producing a numerically solvable system of equations. 
In the modeling stage, components and modules are selected from  libraries and integrated into subsystems. 
Each having its own physical dynamics, these subsystems together can be further interconnected via interface or coupling formulas, see \cite{scholz2013combined} for example.
The result of this approach can be a large, sparse, and nonlinear DAE system, which is typically structured: the dependence between components is stronger within a subsystem, but is weaker between subsystems. Moreover, such a DAE may have a high index.

To solve numerically a DAE, usually derivatives of some equations need to be appended to the original DAE, and an augmented system is solved as a whole. With some index reduction methods \cite{kunkel2004index,Matt93a} or regularization techniques \cite{scholz2013combined,KunM06}, this enlarged system is reduced to a DAE of index 1 or a regularized DAE, respectively, so that a standard DAE numerical solution method can be applied. However, it is not easy to tell which equations are to be differentiated, and how many times exactly. If the numerical method is not chosen properly for a DAE of high index, then the integration can lead to instabilities and non-convergence of this method \cite{scholz2013combined}.

Hence it is desirable to understand the structure of a DAE before a simulation starts on it. As a preprocessing tool, some structural analysis (SA) algorithm is applied to determine the index, number of degrees of freedom (DOF), constraints, and which variables and derivatives need initial values. This preprocess helps give more insight into the underlying structure of a DAE and indicates how to carry out a numerical integration.

The widely used SA method of Pantelides's \cite{Pant88b} is an algorithm that requires graph theory for understanding and implementation. Pryce's \Sigmeth \cite{Pryce2001a} is equivalent to it, for they both produce the same structural index, when applied to first-order systems \cite[Theorem 5.8]{Pryce2001a}. 
This index is an upper bound for the differentiation index, and often they are the same \cite{Pryce2001a}. However, Rei\ss ig {\em et al.} show that the structural index can be arbitrarily high for a family of DAEs of differentiation index 1 \cite{Reissig1999a}. We show that some simple manipulation on equations can make the \Sigmeth report the correct (structural) index 1 on these DAEs \cite{tgn2015aCAS}.

The \Sigmeth can also work on high-order systems. The SA results can help decide how to apply an index reduction algorithm \cite{Matt93a,mckenzie2013AMMCS,Pryce2015c,McKenzie2015UDDs,Pryce2015DDsnew,McKenzie2016a,mckenzie2015a}, perform a regularization process \cite{scholz2013combined}, or design a solution scheme for a Taylor series method \cite{Barrio05b,Barrio06a,nedialkov2005solving,nedialkov2007solving,nedialkov2008solving, nedialkov2016a}.

Although the \Sigmeth succeeds on many problems of practical interest, it can fail---hence Pantelides's algorithm fails as well---on simple, solvable DAEs, producing an identically singular System Jacobian. 

In this article, we investigate the $\Sigma$-method's failures and present two conversion methods that reformulate such a DAE into an equivalent problem with the same solution. 
After each conversion, provided some conditions are satisfied, the value of the signature matrix is guaranteed to decrease. We conjecture that this decrease usually leads to a better formulation of a problem, so that the SA may produce a (generically) nonsingular System Jacobian and hence succeed.

The rest of this article is organized as follows. 
\SCref{SA} summarizes the \Sigmeth theory and the notation we use throughout this article. 
\SCref{fail} describes these SA's failures. 
\SCref{methods} introduces the conversion methods and illustrates them with simple examples.
\SCref{exam} presents two more illustrative examples.
\SCref{conclu} gives conclusions.

\section{Summary of the \Sigmeth.}\label{sc:SA}
We consider DAEs of the general form
\begin{align}\label{eq:maineq}
  f_i(\, t,\, \text{the $x_j$ and derivatives of them}\,) = 0, \quad i=\rnge{1}{n}\;,
\end{align}
where\footnote{The colon notation $\rnge{p}{q}$ for integers $p, q$ denotes either the unordered set or the enumerated list of integers $i$ with $p\le i\le q$, depending on context.} the $x_j(t)$, $j=\rnge{1}{n}$, are state variables that are functions of an independent variable $t$, usually regarded as time. 

We introduce notation that will be used later. For more details, see \cite{Pryce2001a, NedialkovPryce2012a, nedialkov2005solving}. Terms are set in {\sl slanted font} at their defining occurrence.

The \Sigmeth constructs for a DAE \rf{maineq} an $n\times n$ {\sl signature matrix} $\Sigma$, whose $(i,j)$ entry $\sij{i}{j}$ is either an integer $\ge 0$, order of the highest derivative to which variable $x_j$ occurs in equation $f_i$, or $-\infty$ if neither $x_j$ nor its derivatives\footnote{Throughout this article, ``derivatives of $x_j$'' include $x_j$ itself as its $0$th derivative: $x_j^{(l)}=x_j$ if $l=0$.} occur in $f_i$.

A {\sl highest-value transversal} (HVT) of $\Sigma$ is a set $T$ of $n$ positions $(i,j)$ with one entry in each row and each column, such that the sum of these entries is maximized. This sum is the {\sl value of $\Sigma$}, written $\val\Sigma$. If $\val\Sigma$ is finite, then the DAE is {\sl structurally well posed} (SWP); otherwise, $\val\Sigma=-\infty$ and the DAE is {\sl structurally ill posed} (SIP). In the SIP case, there exists no one-to-one correspondence between equations and variables.

We henceforth consider the SWP case.
Using a HVT, we find $2n$ integers $\~c=(c_1,\ldots,c_n)$ and $\~d=(d_1,\ldots,d_n)$ associated with the equations and variables of \rf{maineq}, respectively. These integers satisfy
\begin{equation}\label{eq:cidj}
c_i\ge 0\quad\text{for all $i$;}\qquad d_j-c_i\geq \sigma_{ij} \quad \text{for all $i,j$ with equality on a HVT\;.}
\end{equation}
We refer to such $\~c$ and $\~d$, written as a pair $\offsetpair$, as a {\sl valid offset pair}. 
It is not unique, but there exists a unique elementwise smallest solution $\offsetpair$ of \rf{cidj}, which we refer to as the {\sl canonical offset pair} \cite{Pryce2001a}.

Any valid  $\offsetpair$ can be used to prescribe a stage-by-stage {\sl solution scheme} for solving DAEs by a Taylor series method. The derivatives of the solution are computed in stages 
\begin{equation}\label{eq:stagek}
k=k_d,k_d+1,\ldots,0,1,\ldots, \quad\text{where $k_d=-\max_j{d_j}$}\;.
\end{equation}
At each stage $k$, we solve
\begin{equation}\label{eq:solvefi}
0 = f_i^{(c_i+k)} \qquad \text{for all $i$ such that $c_i+k\geq 0$}
\end{equation}
for derivatives
\begin{equation}\label{eq:forxj}
\phantom{0 =\,\, }x_j^{(d_j+k)} \qquad\text{for all $j$ such that $d_j+k\geq 0$}\;,
\end{equation}
using $x_j^{(< d_j+k)}$, $j=\rnge{1}{n}$, found in the previous stages.
Here $z^{(<r)}$ is a short notation for $z,z',\ldots,z^{(r-1)}$, and $z^{(\le r)}$ includes $z^{(<r)}$ and $z^{(r)}$.

If the solution scheme \rrf{stagek}{forxj} can be carried out for stages $k=\rnge{k_d}{0}$, and the derivatives $x_j^{(\le d_j)}$, $j=\oneton$, can be uniquely determined, then we say the solution scheme and the \Sigmeth\xspace {\sl succeed}. Otherwise they {\sl fail}, in the sense that the Jacobian used to solve \rf{solvefi} at some stage $k\in \rnge{k_d}{0}$ does not have full row rank.

The Jacobian used to solve \rf{solvefi} for stages $k\ge 0$ is called the {\sl System Jacobian} of \rf{maineq}, an $n\times n$ matrix $\Jac\offsetpair=(\Jij{i}{j})$ defined by
\begin{equation}\label{eq:sysjac}
\Jij{i}{j} = \pp{f_i^{(c_i)}}{x_j^{(d_j)}} = \pp{f_i}{x_j^{(d_j-c_i)}} =
\left\{
\begin{array}{lll}
\pp{f_i}{x_j^{(\sigma_{ij})}} & \qquad&\text{if $d_j-c_i=\sigma_{ij},$}\text{ and}\\[1ex]
0 &  \qquad&\text{otherwise\;,}
\end{array}
\right.
\end{equation}
with $i,j=\rnge{1}{n}$. The second ``$=$'' in \rf{sysjac} results from Griewank's Lemma (see later \LEref{griewank}), and the third ``$=$'' follows from \rf{cidj}.

Using the derivatives computed in stages $k=\rnge{k_d}{0}$, we have found a {\sl consistent point}: it is either $\bigl(t,x_1^{(<d_1)},\ldots,x_n^{(<d_n)}\bigr)$, if every $x_j^{(d_j)}$ occurs in a {\em jointly linear} way in every $f_i^{(c_i)}$, or $\bigl(t,x_1^{(\le d_1)},\ldots,x_n^{(\le d_n)}\bigr)$, if some $x_j^{(d_j)}$ occurs nonlinearly in \rf{solvefi} at stage~$k=0$.

Although a different $\offsetpair$ produces a different solution scheme \rrf{stagek}{forxj} and generally a different System Jacobian $\Jac\offsetpair$, all $\Jac$'s nevertheless share the same determinant \cite{nedialkov2005solving}.
If one $\Jac$ is nonsingular---and hence all $\Jac$'s are---at a consistent point, then there exists (locally) a unique solution through this point \cite{Pryce2001a}. 
The SA can now use the canonical $\offsetpair$ to determine the {\sl structural index} and the number of~{\sl DOF}\,:
\begin{align*}
\nu_S &= \max_i c_i+
\left\{
\begin{array}{lll}
1 & \qquad&\text{if $\min_j d_j=0$}\\[1ex]
0&  \qquad&\text{otherwise\;}
\end{array}
\right.\\[1ex]
\text{and}\quad \text{\dof} &= \val{\Sigma} = \sum_{(i,j)\in T}\sij{i}{j} =  \sum_j d_j - \sum_i c_i\;.
\end{align*}
Here ``DOF'' refers to the phrase ``degrees of freedom'', while \dof is the corresponding number.

\begin{example}
We illustrate\footnote{When we present a DAE example, we also present its signature matrix $\Sigma$, the canonical offset pair $\offsetpair$, and the associated System Jacobian $\Jac$.}
the above concepts with the simple pendulum, a DAE of differentiation index 3.
\begin{equation}\label{eq:pend}
\begin{aligned}
0 = f_1 &= x''+x\lam\\
0 = f_2 &= y''+y\lam -g\\
0 = f_3 &= x^2+y^2-L^2
\end{aligned}
\;\;\;\quad
\Sigma = 
\begin{blockarray}{r@{\hspace{3mm}}c@{\hspace{2mm}}c@{\hspace{2mm}}c@{\hspace{2mm}}c}
&  x &   y & \lam  & \s{c_i} & \\
\begin{block}{r @{\hspace{3mm}}[c@{\hspace{2mm}}c@{\hspace{2mm}}c@{\hspace{2mm}}]c}
f_1 & 2^\bullet  &   & 0^\circ & \s0 \\
f_2 &   & 2^\circ  & 0^\bullet & \s0 \\
f_3 & 0^\circ  & 0^\bullet  &  & \s2  \\
\end{block}
\s{d_j}& \s2\phanbull &\s2\phanbull &\s0\phanbull \\[-4ex]
\end{blockarray}
\;\;\;\quad
\Jac = \begin{blockarray}{r@{\hspace{3mm}}c@{\hspace{2mm}}c@{\hspace{2mm}}c}
& x'' &  y'' & \lam \\
\begin{block}{r@{\hspace{3mm}} [c@{\hspace{2mm}}c@{\hspace{2mm}}c]}
f_1 & 1  &    &x  \\
f_2 &   &1    &y  \\
f''_3 & 2x &2y   &  \\
\end{block}
& \\[-4ex]
\end{blockarray} 
\end{equation}
The state variables are $x,y$, and $\lam$;  $G$ is gravity and $L>0$ is the length of the pendulum. There are two HVTs of $\Sigma$, marked with $\bullet$ and $\circ$, respectively. A blank in $\Sigma$ denotes $-\infty$, and a blank in $\Jac$ denotes 0.
The row and column labels in $\Jac$, showing equations and variables differentiated to order $c_i$ and $d_j$, aim to remind the reader of the formula for $\Jac$ in \rf{sysjac}.

Since $\det(\Jac)=-2(x^2+y^2)=-2L^2\neq 0$, the $\Jac$ is nonsingular, and the SA succeeds. 
The derivatives $x'',y'',\lam$ occur in a jointly linear way in \rf{pend}, so a consistent point is
\[
\left(t,x^{(<d_1)},y^{(<d_2)},\lam^{(<d_3)}\right)
= \left(t,x^{(<2)},y^{(<2)},\lam^{(<0)}\right)
= \left(t,x,x',y,y'\right)
\]
that satisfies \rf{solvefi} in stages $k=-2,-1$, that is, $f_3=0$ and $f'_3=0$.
The structural index is $\nu_S=\min_i c_i+1=2+1=3$ (because $\min_j d_j = d_3=0$), which equals the differentiation index. The number of DOF is $\dof=\val\Sigma=\sum_j d_j-\sum_i c_i=4-2=2$. The solution scheme prescribed by the canonical $\offsetpair$ is shown in \TBref{bltriPend}.

\begin{table}[H]
{\normalsize
\[
\begin{array}{*4{r@{\hspace{5mm}}}r}
k &  \text{solve}  &\text{for} &\text{using} & \text{Jacobian} \\
\hline
-2 & f_3     & x,y    &-   & [2x\quad 2y]\\  
-1 & f'_3    & x',y'  &x,y & [2x\quad 2y]  \\
\ge 0 & f_1^{(k)},f^{(k)}_2,f^{(k+2)}_3 & \lam^{(k)},x^{(k+2)},y^{(k+2)} &\lam^{(<k)}, x^{(<k+2)}, y^{(<k+2)} & \Jac \\
\end{array}
\]
\caption{Solution scheme for the simple pendulum DAE.}
\label{tb:bltriPend}
}
\end{table}
\end{example}

\section{Structural analysis's failure.}\label{sc:fail}
We discuss the $\Sigma$-method's failures in this section. Hidden symbolic cancellation is the easiest way that can make the \Sigmeth fail with structurally singular System Jacobian \cite{nedialkov2005solving}; see \SSCref{fail1}. However, some failures of SA can be subtle and obscure, for the System Jacobian  is identically singular but structurally nonsingular. We characterize both failure cases in \SSCref{fail1} and \SSCref{fail2}, respectively.

We use $u\nequiv0$ to mean that $u$ is {\sl generically nonzero} (that is, not identically zero) for all values of the variables occurring in the expressions that define $u$. This $u$ may be a scalar, a vector, or a matrix, depending on context. 
Similarly, we use $\det(\~A)\nequiv0$ to mean that a matrix $\~A$ is {\sl generically nonsingular}, that is, not identically singular.

\subsection{Symbolic cancellation may cause failure.}\label{ssc:fail1}
In the encoding of a DAE, an equation $f_1$ may be, for instance, $x_2+(x_1x_2)'-x_1'x_2$ or $x_1+x_2+\cos^2 x'_1+\sin^2 x'_1$. We say a {\sl symbolic cancellation} occurs in $f_1$, because it simplifies to $x_2+x_1x'_2$ and $x_1+x_2+1$, respectively. That is, $f_1$ does not {\em truly} depend on $x'_1$. However, we note that the problem of detecting such true dependence (which is equivalent to recognizing zero) in any expressions is unsolvable in general \cite{richardson1968}.

Codes like \daets \cite{nedialkov2008solving} and \daesa \cite{NedialkovPryce2012b, NedialkovPryce2012a}, which are implemented through operator overloading and do not perform symbolic simplifications, compute a {\em formal} $\compsij{i}{j}$  instead of a {\em true} one when constructing the signature matrix. For example, both codes would find for $f_1$ above the formal $\compsij{1}{1}=1$ instead of the true $\sij{1}{1}=0$.
By a formal $\compsij{i}{j}$, we mean that $x_j^{(\compsij{i}{j})}$ appears as a highest-order derivative (HOD) in the encoding of an equation $f_i$, while a true $\sij{i}{j}$ means that $f_i$ is not constant with respect to a HOD $x_j^{(\sij{i}{j})}$ and thus truly depends on it---equivalently $\ppin{f_i}{x_j^{(\sij{i}{j})}}\nequiv0$. Obviously~$\compsij{i}{j}\ge\sij{i}{j}$.

For a formally computed $\compSig=(\compsij{i}{j})$, also a valid offset pair $\compoffvec$ is found and a System Jacobian $\compJ$ is derived from $\compoffvec$ and $\compSig$ by \rf{sysjac}. Suppose symbolic cancellations happen in some $f_i$ and make $\compsij{i}{j}>\sij{i}{j}$. Then $f_i$ does not truly depend on $x_j^{(\compsij{i}{j})}$, and $\compJij{i}{j}$ is identically zero by \rf{sysjac}, whether $\widetilde{d}_j-\widetilde{c}_i=\compsij{i}{j}$ holds or not. In this case, $\compJ$ has more identically zero entries than does a $\Jac$ based on the true $\Sigma$ and $\offsetpair$, hence being more likely structurally singular.

Overestimating some $\sij{i}{j}$ of $\Sigma$ may seem dangerous to the SA's success. Fortunately, modern modeling environments usually perform simplifications on problem formulation \cite{maplesimsymdae, Carpanzano1998133, sjolund2011debugging}. They can reduce the occurrence of a structurally singular $\Jac$, when SA is applied. 
Theorems 5.1 and 5.2 in \cite{nedialkov2005solving} also ensure that, if $\val\compSig=\val\Sigma$ and $\det(\Jac)\nequiv 0$, then an offset pair $\compoffvec$ of the formal $\compSig$ is also valid for $\Sigma$, and $\det(\compJ)=\det(\Jac)\nequiv 0$. In this case, such an overestimation would treat some identically zero entries of $\Jac$ as nonzeros and simply make the solution scheme slightly less efficient; see \cite[Examples 5.1 and 5.2]{nedialkov2005solving}. By the same theorems, in the case $\val\compSig>\val\Sigma$, $\compJ$ {\em must} be structurally singular.

\subsection{SA can fail when $\Jac$ is structurally nonsingular.}\label{ssc:fail2}
Hereafter we focus on the case where an identically singular System Jacobian $\Jac$ is structurally nonsingular---that is, there exists a HVT $T$ of $\Sigma$ such that $\Jij{i}{j}\nequiv0$ for all $(i,j)\in T$. We shall simply say ``identically singular'' to refer to this case. 

When $\Jac$ is identically singular, the DAE may be still solvable, but the way its equations are written may not properly reflect its structure.
For example, if the pendulum DAE \rf{pend} $\~f=0$ is equivalently formulated as $\~M\~f=0$ with $\~M$ being a random nonsingular constant $3\times 3$ matrix, then each row of $\Sigma$ is $[2,2,0]$, the canonical offset pair is $(\~c;\~d)=(0,0,0;2,2,0)$, and the resulting $\Jac$ is identically singular \cite{nedialkov2008daets}. 
\begin{example}
We illustrate a failure case with the following DAE\footnote{The original formulation denotes driving functions as $f_1, f_2$.} 
in \cite[p. 23]{BrenanCampbelPetzold}.\begin{equation*}
\begin{alignedat}{3}
0 = f_1 &= x' &&+ ty' &&- h_1(t) \\
0 = f_2 &= x &&+ ty &&- h_2(t)
\end{alignedat}
\qquad\Sigma=
\begin{blockarray}{r@{\hspace{10pt}}c@{\hspace{2mm}}c@{\hspace{10pt}} ll}
&  \phanbull x\phanbull\phanbull  &  y\phanbull& \hspace{-2pt} \s{c_i} \\
\begin{block}{r @{\hspace{10pt}}[c@{\hspace{2mm}}c]@{\hspace{10pt}}ll}
f_1 & 1^\bullet & 1\phanbull   &\s0  \\
f_2 & 0\phanbull & 0^\bullet   &\s1  \\
\end{block}
\s{d_j}& \s1\phanbull &\s1\phanbull  \\[-4ex]
\end{blockarray}
\qquad\Jac = 
\begin{blockarray}{rc@{\hspace{2mm}}c cc}
&  \phanbull x'\phanbull  &  \phanbull y'\phanbull \\
\begin{block}{r @{\hspace{10pt}}[@{\hspace{2mm}}c@{\hspace{2mm}}c]cc}
f_1 & 1 &  t     \\
f'_2 & 1  &  t     \\
\end{block}
&\\[-4ex]
\end{blockarray}
\end{equation*}
The SA fails since $\det(\Jac)\equiv0$. Here $\Jac$ is identically singular but not structurally singular.

One simple fix is to replace $f_1$ by $\newf_1 = -f_1 + f'_2$, 
which results in the problem below; cf.~\cite[Example 5]{Matt93a}.
\begin{equation*}
\begin{aligned}
0 = \newf_1 &= y+h_1(t)-h'_2(t) \\
0 = f_2 &= x + ty - h_2(t)
\end{aligned}
\qquad\newSig=
\begin{blockarray}{r@{\hspace{10pt}}c@{\hspace{2mm}}c@{\hspace{10pt}} ll}
&  \phanbull x\phanbull\phanbull  & y^{\phantom{\bullet}} & \hspace{-2pt} \s{c_i} \\
\begin{block}{r @{\hspace{10pt}}[c@{\hspace{2mm}}c]@{\hspace{10pt}}ll}
\newf_1 &  & 0^\bullet   &\s0  \\
f_2 & 0^\bullet & 0^{\phantom{\bullet}}   &\s0  \\
\end{block}
\s{d_j}& \s0^{\phantom{\bullet}} &\s0^{\phantom{\bullet}} \\[-4ex]
\end{blockarray}
\qquad \newJac = 
\begin{blockarray}{rc@{\hspace{2mm}}c cc}
&  \phanbull x\phanbull  &  \phanbull y\phanbull   \\
\begin{block}{r @{\hspace{10pt}}[@{\hspace{2mm}}c@{\hspace{2mm}}c]cc}
\newf_1 &  &  1     \\
f_2 & 1  &  t     \\
\end{block}
&\\[-4ex]
\end{blockarray}
\end{equation*}
Since $\det(\newJ)=-1$, the SA succeeds. Notice $\val\newSig=0<1=\val\Sigma$. This is a simple illustration of our linear combination method in \SSCref{LC}.

\medskip

Another simple fix is to introduce a variable $z=x+ty$ and eliminate $x$ in $f_1$ and~$f_2$.
\begin{equation*}
\begin{alignedat}{3}
0 = \newf_1 &= -y\ +\ &&z' &&-h_1(t) \\
0 = \newf_2 &=     &&z  &&- h_2(t)
\end{alignedat}
\qquad\newSig=
\begin{blockarray}{r@{\hspace{10pt}}c@{\hspace{2mm}}c@{\hspace{10pt}} ll}
& \phanbull y\phanbull\phanbull  &  z\phanbull & \hspace{-2pt} \s{c_i} \\
\begin{block}{r @{\hspace{10pt}}[c@{\hspace{2mm}}c]@{\hspace{10pt}}ll}
\newf_1 &   0^\bullet  & 1^{\phantom\bullet} &\s0  \\
\newf_2 &  & 0^\bullet    &\s1  \\
\end{block}
\s{d_j}& \s0^{\phantom{\bullet}} &\s1^{\phantom{\bullet}}  \\[-4ex]
\end{blockarray}
\qquad \newJac = 
\begin{blockarray}{rc@{\hspace{2mm}}c@{\hspace{2mm}} cc}
&  \phanbull y\phanbull  &  \phanbull z'\phanbull   \\
\begin{block}{r @{\hspace{10pt}}[c@{\hspace{2mm}}c]cc}
\newf_1 & -1 &  1     \\
\newf'_2 &    &  1     \\
\end{block}
&\\[-4ex]
\end{blockarray}
\end{equation*}
For this resulting DAE, $\det(\newJ)=-1$, and the SA succeeds. After solving for $y$ and $z$, we can obtain $x=z-ty$. This fix also gives $\val\newSig=0<1=\val\Sigma$, and is a simple illustration of our expression substitution method in \SSCref{ES}.
\end{example}

A conjecture in \cite{nedialkov2008daets} attributed the SA's failure to a DAE ``being not sparse enough to reflect its underlying mathematical structure.'' The sparsity refers to occurrence of only a few derivatives in each equation.  However, as we shall see later, decreasing $\val\Sigma$ may be the key to deriving a better problem formulation of a DAE. Our conversion methods aim to do so, and are the main contribution of this article.

\section{Conversion methods.}\label{sc:methods}
We present two conversion methods that attempt to fix SA's failures in a systematic way. The first method is based on replacing an existing equation by a linear combination of some equations and derivatives of them. We call this method the linear combination (LC) method and describe it in \SSCref{LC}. 
The second method is based on substituting newly introduced variables for some expressions and enlarging the system. We call this method the expression substitution (ES) method and describe it in~\SSCref{ES}.

Given a DAE \rf{maineq}, we assume henceforth that $\val\Sigma$ is finite and that a System Jacobian $\Jac$ is identically singular but structurally nonsingular. We also assume that the equations in \rf{maineq} are sufficiently differentiable, so that our methods fit into the $\Sigma$-method theory; see Theorem 4.2 in \cite{Pryce2001a} and \S3 in \cite{nedialkov2005solving}.

After a conversion, we denote the corresponding signature matrix as $\newSig$ and System Jacobian as $\newJac$. If $\val\newSig$ is finite and  $\newJac$ is identically singular still, then we can perform another conversion, using either of the methods, provided the corresponding conditions are satisfied.

Suppose a sequence of conversions produces a solvable DAE with $\val\newSig\ge0$ and a generically nonsingular $\newJ$. Given the fact that each conversion reduces the value of the signature matrix by at least one, the total number of conversions does not exceed the value of the original signature matrix.

If the resulting system is SIP after a conversion, that is, $\val\newSig=-\infty$, then we say the original DAE is {\sl ill posed}.

\subsection{Linear combination method.}\label{ssc:LC}

Let $\vecu=[u_1,\ldots,u_n]^T\nequiv \~0$ be a nonzero $n$-vector function in the cokernel of $\Jac$, that is, $\vecu\in\coker\Jac$ or equivalently $\Jac^T \vecu=\~0$. We consider $\Jac$ and $\vecu$ as functions of $t$ and derivatives of the $x_j(t)$'s, $j=\oneton$. 

For convenience, denote
\begin{align}\label{eq:truehod}
\hoder{x_j}{\omega} &=
\left\{
\begin{array}{l}
\text{order of the highest derivative to which $x_j$ occurs in $\omega$; or}\\[1ex]
-\infty\,\,\text{ if $x_j$ does not occur in $\omega$\;.}
\end{array}
\right.
\end{align}
Here $\omega$ can be a scalar, a vector, or a matrix, depending on context.
This notation is a generalization of the $(i,j)$ entry of $\Sigma$: $\sij{i}{j}=\hoder{x_j}{f_i}$.

\begin{lemma}\label{le:griewank}
{\bf\em (Griewank's Lemma)}\cite{nedialkov2007solving}
Let $w$ be a function of $t$, the $x_j(t)$, $j=\oneton$, and derivatives of them. Denote $w^{(p)}=\text{d}^p w/\text{d}t^p$, where $p\ge 0$. If $\hoder{x_j}{w}\le q$, then
\begin{equation}\label{eq:griewank}
\pp{w}{x_j^{(q)}} = \pp{w'}{x_j^{(q+1)}} = \cdots = \pp{w^{(p)}}{x_j^{(q+p)}}\;.
\end{equation}
\end{lemma}

Denote
\begin{align}\label{eq:LCdef}
\eqsetI = \{ \, i \mid  u_i\nequiv 0 \,\}, \quad \LCconst = \min_{i \in \eqsetI} c_i, \qaq \nzset=\setbg{\,i\in \eqsetI \mid\, c_i=\LCconst }\;.
\end{align}
We prove two preliminary lemmas before the main \THref{LCConv}, on which the LC method is based.

\begin{lemma}\label{le:LCConv}
Assume that $\vecu\in\coker\Jac$ and $\vecu\nequiv\~0$. If
\begin{equation}\label{eq:LCcond}
\hoder{x_j}{\vecu}< d_j-\LCconst, \quad \text{for all $j=\oneton$}\;,
\end{equation}
then $\hoder{x_j}{\newf}< d_j-\LCconst\ $  for all $j=\oneton$, where
\begin{equation}\label{eq:LCnewf}
\newf =  \sum_{i\in \eqsetI} u_i f_i^{\left(c_i-\LCconst\right)}\;.
\end{equation}
\end{lemma}

\begin{proof}
The formula for $\LCconst$ gives $c_i-\LCconst \ge 0$ for all $i\in \eqsetI$. By \rf{cidj}, $\hoder{x_j}{f_i}=\sij{i}{j}\le d_j-c_i$. Applying Griewank's Lemma \rf{griewank} to \rf{sysjac}  with $w=f_i$ and $q=c_i-\LCconst$ yields
\begin{align}\label{eq:gl}
\Jac_{ij} 
= \frac{\partial f_i}{\partial x_j^{(d_j-c_i)}} 
= \frac{\partial f_i^{(c_i-\LCconst)}}{\partial x_j^{(d_j-c_i+c_i-\LCconst)}} 
= \frac{\partial f_i^{(c_i-\LCconst)}}{\partial x_j^{(d_j-\LCconst)}} \qquad \text{for $i\in \eqsetI$ and all $j=\oneton$}\;.
\end{align}
This shows that such an $f_i^{(c_i-\LCconst)}$ depends on $x_j^{(\le d_j-\LCconst)}$ only. Then for all $j=\oneton$,
\begin{alignat}{3}
\hspace{6mm}
\frac{\partial \newf}{\partial x_j^{(d_j-\LCconst)}} 
&= \frac{\partial
\left(  \sum_{i \in \eqsetI} u_i  f_i^{(c_i-\LCconst)}  \right)
}
{\partial x_j^{(d_j-\LCconst)}}
&\hspace{10mm}&\text{by the definition of $\newf$ in \rf{LCnewf}}  
\nonumber\\
&= \sum_{i \in \eqsetI} u_i \frac{\partial f_i^{(c_i-\LCconst)}}{\partial x_j^{(d_j-\LCconst)}} 
=\sum_{i \in \eqsetI} u_i \Jac_{ij}
&&\text{by \rf{LCcond} and then \rf{gl} }
\label{eq:deriveLC}\\
&
= (\Jac^T \vecu)_j = 0
&& \text{since $\vecu\in\coker\Jac$
}\;. 
\nonumber
\end{alignat}
Hence $\newf$ depends on $x_j^{(<d_j-\LCconst)}$ only, for all $j$---this results in the inequality in \rf{LCnewf}.\qed
\end{proof}

\begin{lemma}\label{le:sigred}
Assume that an $n\times n$ signature matrix $\Sigma$ has a finite $\val\Sigma$ and a valid offset pair $\offsetpair$.
Given a row of index $\indxk$, if we replace in row $\indxk$ all entries $\sij{\indxk}{j}$ by $\newsij{\indxk}{j}<d_j-c_\indxk$,  then $\val{\newSig} < \val{\Sigma}$, where $\newSig$ is the resulting signature matrix.
\end{lemma}

\begin{proof}
Since $\newsij{\indxk}{j} < d_j-c_\indxk$ for all $j$, the intersection of a HVT $\newT$ of $\newSig$ with row $\indxk$ is a position $(\indxk,r)$ with $\newsij{\indxk}{r} < d_r-c_\indxk$. Then
\begin{equation*}
\hspace{9mm}
\val{\newSig} 
= \sum_{(i,j)\in T} \newsij{i}{j}
= \newsij{l}{r} + \sum_{(i,j)\in T\setminus\set{(l,r)}} \sij{i}{j}
< \sum_j d_j - \sum_i c_i = \val{\Sigma}\;.
\hspace{9mm}\qed
\end{equation*}
\end{proof}

The LC method is based on the following theorem.
\begin{theorem}\label{th:LCConv}
Let $\eqsetI$, $\LCconst$, and $\nzset$ be as defined in \rf{LCdef}. If we replace an equation $f_\indxk$, $\indxk\in\nzset$, by $\newf$ in \rf{LCnewf}, then  $\val\newSig<\val\Sigma$, where $\newSig$ is the signature matrix of the resulting~DAE.
\end{theorem}
\begin{proof}
By \LEref{LCConv}, such a replacement results in $\newsij{\indxk}{j}=\hoder{x_j}{\newf_\indxk}< d_j-c_\indxk$ for all $j=\oneton$. Immediate from \LEref{sigred} is $\val\newSig<\val\Sigma$.\qed
\end{proof}

Usually we write $\newf$ as $\newf_\indxk$ in the resulting DAE.

We call \rf{LCcond} the LC condition, which is merely {\em sufficient} for the strict decrease: 
if \rf{LCcond} becomes $\hoder{x_j}{\vecu}\le d_j-\LCconst$ for all $j=\oneton$ with equality for some $j$, then we can achieve only $\val{\newSig}\le\val{\Sigma}$, while the strict ``$<$'' may not hold.

\begin{example}
We illustrate the LC method with the following simple example:
\begin{equation*}
\begin{alignedat}{4}
0= f_1 &= -&&x'_1 + x_3 \\
0= f_2 &= -&&x'_2+ x_4
\end{alignedat}
\qquad\qquad
\begin{alignedat}{4}
0= f_3 &= &&x_1 x_2+g_1(t) \\
0= f_4 &= &&x_1 x_4 + x_2x_3 + x_1+x_2&&+g_2(t)\;,
\end{alignedat}
\end{equation*}
where $g_1$ and $g_2$ are driving functions.
\begin{align*}
\Sigma=
\begin{blockarray}{r@{\hspace{10pt}}*4{c@{\hspace{2mm}}}@{\hspace{1mm}} ll}
&  x_1\phanbull &   x_2\phanbull &  x_3\phanbull &  x_4\phanbull & \s{c_i} \\
\begin{block}{r @{\hspace{10pt}}[*4{c@{\hspace{2mm}}}]@{\hspace{1mm}}ll}
f_1 & 1^\bullet  &  & 0^{\phantom\bullet}&     &\s0  \\
f_2 &   &1^{\phantom\bullet}  & & 0^\bullet    &\s0  \\
f_3 & 0^{\phantom\bullet}  &0^\bullet  & &     &\s1  \\
f_4 & \lgo^{\phantom\bullet} &\lgo^{\phantom\bullet}  & 0^\bullet & 0^{\phantom\bullet}    &\s0  \\
\end{block}
 \s{d_j}& \s1\phanbull &\s1\phanbull &\s0\phanbull  &\s0\phanbull \\[-4ex]
 \end{blockarray}
\qquad\qquad \Jac = 
\begin{blockarray}{r *4{c@{\hspace{2mm}}} cc}
&  x'_1\phanbull &   x'_2\phanbull &  x_3\phanbull &  x_4\phanbull \\
\begin{block}{r @{\hspace{10pt}}[*4{c@{\hspace{2mm}}}]cc}
f_1 & -1  &  & 1&       \\
f_2 &   &-1  & & 1      \\
f'_3 & x_2  & x_1  & &\\
f_4 &   &  & x_2 & x_1     \\
\end{block}
& \\[-4ex]
\end{blockarray}
\end{align*}
A shaded entry $\OK{\sij{i}{j}}$ in $\Sigma$ denotes a position $(i,j)$ where $d_j-c_i>\sigma_{ij}\ge 0$ and hence $\Jij{i}{j}\equiv0$ by the formula \rf{sysjac} for $\Jac$. The SA fails here since $\det(\Jac)\equiv0$.

We choose $\vecu = \bigl( x_2, x_1,1,-1\bigr)^T\in\coker\Jac$. Then \rf{LCdef} becomes
\begin{align*}
\eqsetI = \setbg{\,i \mid\, u_i\nequiv 0\,} = \setbg{\rnge{1}{4}},\quad 
\underC = \min_{i\in \eqsetI}c_i=0, \quad
\nzset = \setbg{\, i \in \eqsetI \mid\, c_i=\underC} = \setbg{1,2,4}\;.
\end{align*}
Checking the condition \rf{LCcond} is not difficult; for example, $\hoder{x_1}{\vecu} =0 < 1 = d_1-\underC$.

We pick $\indxk=4\in \nzset$ (we shall reason why this choice is desirable) and replace $f_4$~by 
\begin{align*}
\newf_4 = \sum_{i\in \eqsetI} u_i f_i^{(c_i-\underC)} 
= x_2 f_1 + x_1 f_2 + f'_3 - f_4
=   -x_1-x_2 +g'_1(t) - g_2(t)\;.
\end{align*}
The resulting DAE is $0=(f_1,f_2,f_3,\newf_4)$.
\begin{align*}
\newSig=
\begin{blockarray}{r@{\hspace{10pt}}*4{c@{\hspace{2mm}}}@{\hspace{1mm}} ll}
&  x_1\phanbull &   x_2\phanbull &  x_3\phanbull &  x_4\phanbull & \s{c_i}  \\
\begin{block}{r @{\hspace{10pt}}[*4{c@{\hspace{2mm}}}]@{\hspace{1mm}}ll}
f_1 & 1^{\phantom\bullet}  &  & 0^\bullet&     &\s0  \\
f_2 &   &1^{\phantom\bullet}  & & 0^\bullet    &\s0  \\
f_3 & 0^{\phantom\bullet}  &0^\bullet  & &     &\s1  \\
\newf_4 & 0^\bullet  &0^{\phantom\bullet}  & &     &\s1  \\
\end{block}
\s{d_j}& \s1\phanbull &\s1\phanbull &\s0\phanbull  &\s0\phanbull \\[-4ex]
\end{blockarray}
\qquad\qquad\newJ = 
\begin{blockarray}{r *4{c@{\hspace{2mm}}} cc}
&  x'_1\phanbull &   x'_2\phanbull &  x_3\phanbull &  x_4\phanbull  \\
\begin{block}{r @{\hspace{10pt}}[*4{c@{\hspace{2mm}}}]cc}
f_1 & -1  &  & 1&     \\
f_2 &   &-1  & & 1    \\
f'_3 & x_2 & x_1 & &     \\
\newf'_4 &-1  &-1  & &    \\
\end{block}
&\\[-4ex]
\end{blockarray}
\end{align*}
Now  $\val\newSig=0<1=\val\Sigma$. The SA succeeds at all points where
\[\det(\newJ)=x_2-x_1\neq 0\;.\]
\end{example}

From \rf{LCdef} and \rf{LCnewf}, we can recover the replaced equation $f_l$ by 
\begin{equation}\label{eq:LCnonzero}
\textstyle{f_\indxk = \left( \newf_\indxk - \sum_{i\in\eqsetI\setminus\{\indxk\}} u_i f_i^{(c_i-\underC)}   \right)\big/u_l }  \;.
\end{equation}
Provided $u_l\neq 0$ for all $t$ in the interval of interest, it is not difficult to show that the original DAE and the resulting one have the same solution, if there exists one. From our experience, it is desirable to choose a row index $\indxk\in\nzset$ such that $u_\indxk$ could be an expression that  {\em never} becomes zero. For example, $u_\indxk$ is a nonzero constant, $x_1^2+1$, or $2+\cos x_2$. Such a choice of $\indxk$ guarantees that the resulting DAE is ``equivalent'' to the original DAE, in the sense that they {\em always} have the same solution if there exists one. The reader is referred to  \cite[\S 5.3]{tgn2015aCAS} for details on the equivalence of DAEs.

Since determining whether an expression is identically zero is unsolvable in general \cite{richardson1968}, we consider a (nonzero) constant $u_\indxk$ as the most preferable choice among all $\indxk\in\nzset$, and derive a set $\newL$ that contains all $\indxk$ for such $u_\indxk$: $\newL=\setbg{ \indxk\in\nzset \mid \text{$u_\indxk$ is constant}}$.

We summarize the steps of the LC method.
\begin{enumerate}[1)]
\item Obtain a symbolic form of $\Jac$.
\item Compute a $\vecu\in\coker\Jac$.
\item Derive $\eqsetI$, $\LCconst$, and $\nzset$ as defined in \rf{LCdef}.
\item Check condition \rf{LCcond}. If it is not satisfied, then set $\nzset\as\emptyset$ to mean that the LC method is not applicable; otherwise proceed to the next step.
\item $\newL \as \setbg{ \indxk\in\nzset \mid \text{$u_\indxk$ is constant}}$. If $\newL\neq\emptyset$, then choose an $\indxk\in\newL$; otherwise an $\indxk \in\nzset$.
\item Replace $f_\indxk$ by $\newf_\indxk=\newf$ as defined in \rf{LCnewf}.
\end{enumerate}

The sets $\nzset$ and $\newL$ are used to decide a desirable conversion method; see \TBref{LCES}.

We show below that the LC method cannot fix the following (artificially constructed) DAE \rf{ESexampleDAE} because the condition \rf{LCcond} is not satisfied.
\begin{example}
Consider
\begin{equation}\label{eq:ESexampleDAE}
\begin{aligned}
0 = f_1 &= x_1 + e^{-x_1'-x_2x_2''} + h_1(t) \\
0 = f_2 &= x_1 + x_2x_2' + x_2^2+h_2(t)\;.
\end{aligned}
\end{equation}
\begin{equation*}
\Sigma=
\begin{blockarray}{r@{\hspace{2mm}}c@{\hspace{2mm}}c@{\hspace{2mm}}l}
&  x_1\phanbull &  x_2\phanbull & \s{c_i} \\
\begin{block}{r @{\hspace{2mm}}[c@{\hspace{2mm}}c]@{\hspace{2mm}}l}
f_1 & 1^\bullet & 2^{\phantom\bullet}   &\s0  \\
f_2 & 0^{\phantom\bullet} & 1^\bullet   &\s1  \\
\end{block}
\s{d_j}& \s1\phanbull  &\s2\phanbull  \\[-4ex]
\end{blockarray}
\quad \Jac = 
\begin{blockarray}{r@{\hspace{2mm}}c@{\hspace{2mm}}c}
&  x'_1 &  x''_2   \\
\begin{block}{r @{\hspace{2mm}}[c@{\hspace{2mm}}c]}
f_1 & -\alpha &  -\alpha x_2     \\
f'_2 & 1  &  x_2     \\
\end{block}
\end{blockarray}
\end{equation*}
Here $h_1$ and $h_2$ are given driving functions, and $\alpha=e^{-x_1'-x_2x_2''}$.
Obviously $\det(\Jac)\equiv0$ and the SA fails.

Choose $\vecu = \bigl( \alpha^{-1},1  \bigr)^T=\bigl( e^{x'_1+x_2x''_2},1 \bigr)^T \in\coker\Jac$.
Then \rf{LCdef} becomes
\begin{align*}
\eqsetI = \setbg{\,i \mid\, u_i\nequiv 0\,} = \setbg{1,2},\quad \underC = \min_{i\in \eqsetI}c_i=0, \,\text{ and }\,\nzset = \setbg{\,i\in \eqsetI \mid\, c_i=\underC} = \setbg{1}\;.
\end{align*}
Obviously, $\hoder{x_j}{\vecu}=d_j-\underC$, for $j=1,2$,
violates \rf{LCcond}. Choosing $\indxk=1\in\nzset$ and replacing $f_1$ by 
\[
\newf_1 = u_1 f_1 + u_2 f'_2 = \beta +x'_1+x_2x''_2+(x'_2)^2 +2x_2 x'_2+h'_2(t)
\]
results in the DAE $0=\left(\newf_1, f_2\right)$. 
Here $\beta=\alpha^{-1}(x_1 + h_1(t)) + 1$.
\begin{align*}
\newSig=
\begin{blockarray}{r@{\hspace{10pt}}c@{\hspace{2mm}}cll}
&  \phanbull x_1\phanbull\phanbull &  x_2\phanbull & \s{c_i} \\
\begin{block}{r @{\hspace{10pt}}[c@{\hspace{2mm}}c]ll}
\newf_1 & 1^\bullet & 2^{\phantom{\bullet}}   &\s0  \\
f_2 & 0^{\phantom{\bullet}} & 1^\bullet   &\s1  \\
\end{block}
\s{d_j}& \s1\phanbull  &\s2\phanbull     \\[-4ex]
\end{blockarray}
\qquad\qquad\newJ = 
\begin{blockarray}{rc@{\hspace{2mm}}c cc}
&  x'_1  &  x''_2   \\
\begin{block}{r @{\hspace{10pt}}[c@{\hspace{2mm}}c]cc}
\newf_1 & \beta  & \beta x_2    \\
f'_2 & 1 & x_2   \\
\end{block}
&\\[-4ex]
\end{blockarray}
\end{align*}
The SA fails still, since $\det(\newJac)\equiv0$.
Now $\val\newSig=\val\Sigma=2$.

We shall show in \EXref{ESexam} that the ES method can fix \rf{ESexampleDAE}.
\end{example}

\subsection{Expression substitution method.}\label{ssc:ES}
Let $\vecv=[v_1,\ldots,v_n]^T\nequiv\~0$ be a nonzero $n$-vector function in the kernel of $\Jac$, that is, $\vecv\in\ker\Jac$, or equivalently $\Jac \vecv = \~0$. Denote
\begin{equation}\label{eq:ESdef}
\begin{alignedat}{5}
\nzsetES &= \setbg{\, j\,\mid\, v_j\nequiv 0\,},\quad s = |\nzsetES|\;, \\
\eqsetIES &= \setbg{\, i\,\mid\, d_j-c_i=\sij{i}{j} \text{ for some } j\in \nzsetES },\qaq \overC = \max_{i\in \eqsetIES} c_i\;.
\end{alignedat}
\end{equation}

We choose an $l\in\nzsetES$, and introduce $s-1$ new variables
\begin{align}\label{eq:gjgk}
y_j &= 
x_j^{(d_j-\overC)} - \frac{v_j}{v_\indxk} \cdot x_\indxk^{(d_\indxk-\overC)} \qquad\text{for all $j\in\Jnol$}\;.
\end{align}
In each $f_i$, we
\begin{equation}\label{eq:ESsubs}
\begin{aligned}
\text{replace every $x_j^{(\sij{i}{j})}
= x_j^{(d_j-c_i)}$ with $j\in\Jnol$ by 
$\Big( y_j + \frac{v_j}{v_\indxk}\cdot x_\indxk^{(d_\indxk-\overC)} \Big)^{(\overC-c_i)}$}\;.
\end{aligned}
\end{equation}
From the formula \rf{ESdef} for $\eqsetIES$, these replacements (or substitutions) occur only in $f_i$'s with $i\in\eqsetIES$, because at least one equality $d_j-c_i=\sij{i}{j}$ must hold for some $j\in\nzsetES$. Hence they use the fact that, for such an $x_j^{(\sij{i}{j})}$, $i\in\eqsetIES$ and $j\in\Jnol$,
\[
x_j^{(\sij{i}{j})} = x_j^{(d_j-c_i)} 
= \left(x_j^{(d_j-\overC)}\right)^{(\overC-c_i)} 
= \left( y_j+\frac{v_j}{v_\indxk}\cdot x_\indxk^{(d_\indxk-\overC)}  \right)^{(\overC-c_i)}\;.
\]

After the replacements, denote each equation by $\newf_i$ (for all  $i\notin \eqsetIES$, $\newf_i$ and $f_i$ are  the same).
Equivalent to \rf{gjgk} are $s-1$ equations
\begin{align}\label{eq:gj}
0 = g_j &= -y_j + x_j^{(d_j-\overC)} - \frac{v_j}{v_\indxk}\cdot x_\indxk^{(d_\indxk-\overC)} \quad \text{for all $j\in\Jnol$}
\end{align}
that prescribe the substitutions in \rf{ESsubs}. 
Appending \rf{gj} to the $\newf_i$'s results in an enlarged DAE consisting of 
\begin{alignat*}{5}
\text{equations} &\quad&&\text{$0=\left(\newf_1,\ldots,\newf_n\right)$}\quad&&\text{and}\quad&&0=g_j &&\,\,\,\,\text{for all $j\in\Jnol$} \\
\text{in variables}&\quad&&\phantom{(}x_1,\ldots,x_n \quad&&\text{and}\quad &&y_j &&\,\,\,\,\text{for all $j\in\Jnol$}\;.
\end{alignat*}

The ES method is based on the following theorem.
\begin{theorem}\label{th:ES}
Let $\nzsetES$, $s$, $\eqsetIES$, and $\overC$ be as defined in \rf{ESdef}. Assume that
\begin{equation}\label{eq:EScond}
\begin{aligned}
\hoder{x_j}{\vecv}  \casemod{ll}{
<d_j-\overC &\,\,\, \text{if } j\in \nzsetES \\[1ex]
\le d_j-\overC &\,\,\, \text{otherwise\;,}} \qaq d_j - \overC \ge 0 \quad \text{for all $j\in\nzsetES$}\;. 
\end{aligned}
\end{equation}
For any $\indxk\in\nzsetES$, if we
\begin{enumerate}[1)]
\item introduce $s-1$ new variables $x_j$, $j\in\Jnol$, as defined in \rf{gjgk},
\item perform substitutions in $f_i$, for all $i=\oneton$, by \rf{ESsubs}, and 
\item append $s-1$ equations $g_j$, $j\in\Jnol$, as defined in \rf{gj},
\end{enumerate}
then $\val{\newSig} < \val{\Sigma}$, where $\newSig$ is the signature matrix of the resulting DAE.
\end{theorem}

We call \rf{EScond} the ES conditions, which are again {\em sufficient} for $\val{\newSig} < \val{\Sigma}$.
\begin{example}\label{ex:ESexam}
We illustrate the ES method  on the DAE \rf{ESexampleDAE}.

Suppose we choose $\vecv=[x_2,-1]^T\in\ker\Jac$. Then \rf{ESdef} becomes
\begin{equation*}
\begin{aligned}
\nzsetES =\setbg{1,2},\quad s=|\nzsetES|=2,\quad
\eqsetIES =\setbg{1,2},\quad \text{and}\quad \ESconst=\max_{i\in \eqsetIES} c_i = c_2 = 1\;.
\end{aligned}
\end{equation*}
We can apply the ES method as the conditions \rf{EScond} hold:
\begin{alignat*}{8}
\hoder{x_1}{\vecv} &= &-\infty\ &&\le 1-1-1 &= d_1-\ESconst-1, &\quad d_1 - \ESconst = 1-1  \ge 0\;,\\
\hoder{x_2}{\vecv} &= &0\,\,\, &&\le  2-1-1 &= d_2-\ESconst-1, &\quad d_2 - \ESconst = 2-1  \ge 0\;.
\end{alignat*}

We choose $\indxk=2\in \nzsetES$. Now $\Jnol =\setbg{1}$. Using \rf{gjgk} and \rf{gj}, we introduce for $x_1$ a new variable
\begin{align*}
y_1 = x_1^{(d_1-\ESconst)} - \frac{v_1}{v_2}\cdot x_2^{(d_2-\ESconst)}  = x_1^{(1-1)} - \frac{x_2}{(-1)}\cdot x_2^{(2-1)} = x_1 + x_2x'_2\;,
\end{align*}
and append the equation $0 = g_1 = -y_1 + x_1 + x_2x'_2$.
Then we replace $x'_1$ by $(y_1 - x_2 x'_2)'$ in $f_1$ to obtain $\newf_1$, and replace  $x_1$ by $y_1 - x_2 x'_2$ in $f_2$ to obtain $\newf_2$.
The resulting DAE and its SA results are shown below.
\begin{equation*}
\begin{alignedat}{3}
0 = \newf_1 
&= x_1 + e^{-y'_1 +x'^2_2} + h_1(t) \\
0 = \newf_2 
&= y_1+x_2^2+h_2(t) \\
0 = g_1 &= -y_1 + x_1 + x_2x'_2
\end{alignedat}
\end{equation*}
\begin{equation*}
\newSig =
\begin{blockarray}{r@{\hspace{10pt}}c@{\hspace{2mm}}c@{\hspace{2mm}}c@{\hspace{10pt}} ll}
&  x_1\phanbull &  x_2\phanbull  &  y_1\phanbull  & \hspace{-2pt}\s{c_i} \\
\begin{block}{r @{\hspace{10pt}}[c@{\hspace{2mm}}c@{\hspace{2mm}}c]@{\hspace{10pt}}ll}
\newf_1 & 0\phanbull & 1\phanbull  & 1^\bullet    &\s0  \\
\newf_2 &  & 0^\bullet  & 0\phanbull    &\s1  \\
g_1 & 0^\bullet & 1\phanbull  & \lgo\phanbull     &\s0  \\
\end{block}
\s{d_j} & \s0\phanbull & \s1\phanbull &\s1\phanbull &\\[-4ex]
\end{blockarray}\qquad
\newJ = 
\begin{blockarray}{r@{\hspace{2mm}}c@{\hspace{2mm}}c@{\hspace{2mm}}c@{\hspace{2mm}}c@{\hspace{2mm}}c}
&  x_1\phanbull &  x'_2\phanbull  &  y'_1\phanbull   \\
\begin{block}{r @{\hspace{10pt}}[c@{\hspace{2mm}}c@{\hspace{2mm}}c]cc}
\newf_1 & 1 & 2x'_2\gamma  &  -\gamma   \\
\newf'_2 &  & 2x_2  & 1\\
g_1 & 1 & x_2  &  \\
\end{block}
\end{blockarray}
\end{equation*}
Here $\gamma = e^{-y'_1 +x'^2_2}$. Now $\val\newSig=1<2=\val\Sigma$.
The SA succeeds at all points where
$\det(\newJ)=2\gamma(x_2+x'_2)-x_2\neq 0$.
\end{example}

We prove a lemma related to \THref{ES}, using the following assumptions for the sake of the proof.
\begin{enumerate}[(a)]
\item Without loss of generality, we assume that the entries $v_j\nequiv 0$ are in the first $s$ positions of $\vecv$, that is,
$\vecv = [v_1,\ldots, v_s,0,\ldots ,0]^T$.
Then $\nzsetES= \setrnge{1}{s}$  in \rf{ESdef}.
\item We introduce one more variable $y_\indxk=x_\indxk^{(d_\indxk-\ESconst)}$ for the chosen $\indxk\in\nzsetES$, and append correspondingly one more equation $0=g_\indxk = -y_\indxk + x_\indxk^{(d_\indxk-\ESconst)}$.
\end{enumerate}

\begin{lemma}\label{le:ESnewSig}
Let $(\~c;\~d)=(c_1,\ldots,c_n; d_1,\ldots,d_n)$ be a valid offset pair of $\Sigma$. Let $\~{\newc}$ and $\~{\newd}$ be the two $(n+s)$-vectors defined as
\begin{align}\label{eq:newcdES}
\newd_j = \casemod{ll}
{d_j &\quad\text{if $j=\oneton$}\\[1ex]
\ESconst  &\quad\text{if $j=n+\oneton+s$}} 
\qaq
\newc_i =\casemod{ll}
{c_i &\quad\text{if $i=\oneton$}\\[1ex]
\ESconst &\quad\text{if $i=n+\oneton+s$}\;,}\ 
\end{align}
where $\ESconst$ is as defined \rf{ESdef}. Then the signature matrix $\newSig$ of the resulting DAE from the ES method has the form in \FGref{ESnewSig}.

\begin{figure}[th]
\begin{align*}
\renewcommand{\arraystretch}{1.5}
\begin{blockarray}{r*7{@{\hspace{2mm}}c}@{\hspace{10pt}}l}
& x_1\,\,\,\cdots\,\,\, x_{\indxk-1}  & x_{\indxk}  & x_{\indxk+1}\,\,\,\cdots\,\,\, x_s  & x_{s+1}\,\,\,\cdots\,\,\,x_n & y_1\,\,\,\cdots\,\,\, y_{\indxk-1}   & y_\indxk   & y_{\indxk+1}\,\,\,\cdots\,\,\, y_s \,\,\, & \newc_i \\
\begin{block}{r@{\hspace{2mm}} [ 
@{\hspace{2mm}}l
@{\hspace{2mm}}c
@{\hspace{2mm}}l
@{\hspace{2mm}}|
@{\hspace{2mm}}c
@{\hspace{2mm}}|
@{\hspace{2mm}}c
@{\hspace{2mm}}c
@{\hspace{2mm}}c@{\hspace{-1mm}}]
 @{\hspace{10pt}}c} 
\newf_1 &&&&&&-\infty && c_1     \\
\vdots &   & \HugeLess &  & \HugeLe  & \HugeLe &\vdots & \HugeLe & \vdots\\
\newf_n &&&&&&-\infty & &  c_n \\  
\cline{2-8}
g_1 & \clsp = \qquad \multirow{2}{0.5cm}{\parbox{12pt}{\HugeLess}}   &  =  &\qquad\multirow{3}{0.5cm}{\LargeLess} & \multirow{2}{0.5cm}{\parbox{12pt}{\HugeLe}}& \mc{1}{l}{\phantom{==}$0$} &&&  \ESconst\\
\vdots  & \clsp\phantom{===}\ddots  & \vdots &&& \phantom{==}\ddots & \mc{2}{c}{\HugeInf} & \vdots \\
g_{\indxk} &   & = &  & -\infty \cdots -\infty && 0 && \ESconst\\
\vdots & \qquad\multirow{1}{0.5cm}{\LargeLess} & \vdots  & \clsp\phantom{==}\ddots  & \multirow{2}{0.5cm}{\parbox{12pt}{\HugeLe}} & \HugeInf & & \ddots\phantom{==} & \vdots\\
g_s &  & = & \multirow{-3}{0.5cm}{\LargeLess}  \clsp\phantom{===}=&  &   & & \mc{1}{r}{$0$\phantom{===}} &\ESconst\\ 
\end{block}
\newd_j & d_1\,\,\,\cdots\,\,\, d_{\indxk-1}  & d_{\indxk}  & d_{\indxk+1}\,\,\,\cdots\,\,\, d_s  & d_{s+1}\,\,\,\cdots\,\,\,d_n 
& \ESconst\,\,\,\,\,\,\,\,\cdots\,\,\,\,\,\,\,\, \ESconst\,\,   & \ESconst   & \ESconst\,\,\,\,\,\,\cdots\,\,\,\,\,\, \ESconst\\[-5ex]
\end{blockarray}
\end{align*}
\caption{\label{fg:ESnewSig} The form of $\newSig$ for the resulting DAE by the ES method. The $<$, $\le$, and $=$ mean the relations between $\newsig_{ij}$ and $\newd_j-\newc_i$, respectively. For instance, every $\newsig_{ij}$ whose $(i,j)$ position is in the region marked with ``$\le$'' is $\le\newd_j-\newc_i$.}
\end{figure}
\end{lemma}

The proof of this lemma is rather technical, so we present it in \APref{proof}. Using \LEref{ESnewSig}, we prove \THref{ES}.
\begin{proof}
Let $\newT$ be a HVT of $\newSig$.
By \LEref{ESnewSig},
\begin{alignat*}{3}
\val{\newSig} &= \sum_{(i,j)\in \newT} \newsij{i}{j} \le \sum_{(i,j)\in \newT}(\newd_j-\newc_i) &\qquad& \text{since $\newd_j-\newc_i\ge \newsij{i}{j}$ for all $i,j$} 
\\&=\sum_{j=1}^{n+s} \newd_j - \sum_{i=1}^{n+s} \newc_i
= \sum_{j=1}^n d_j + s\ESconst - \sum_{i=1}^n c_i - s\ESconst &&\text{by \rf{newcdES}}
\\&= \sum_{j=1}^{n} d_j - \sum_{i=1}^{n} c_i = \val{\Sigma}\;.
\end{alignat*}
We assert $\val{\newSig}<\val{\Sigma}$, and show that an equality leads to a contradiction.

Assume that $\val{\newSig}=\val{\Sigma}$. Then there exists a transversal $\newT$  of $\newSig$ such that 
\begin{align}\label{eq:ass2}
\newd_j-\newc_i = \newsig_{ij}>-\infty \qquad\qquad \text{for all $(i,j)\in \newT$}\;.
\end{align}
Consider $(i_1,1), \ldots,  (i_s,s) \in \newT$ for the first $s$ columns. 
Since the $y_\indxk$ column has only one finite entry $\newsig_{n+\indxk,n+\indxk}=0$, position $(n+\indxk,n+\indxk)$ is in $\newT$, and thus row numbers 
$i_1,\ldots,i_s$ can only take values among 
\begin{equation*}
1,\,2,\,\ldots ,\,n,\, n+1,\,\ldots,\, n+\indxk-1,\, n+\indxk+1, \,\ldots ,\,n+s\;.
\end{equation*}
Here only $s-1$ numbers are greater than $n$, so at least one of them is among $\rnge{1}{n}$.
In other words, there exists a position $(r,j)\in\newT$ with $1\le r\le n$ and $1\le j\le s$ in the ``$<$'' region in \FGref{ESnewSig}. 
Hence $\newd_j- \newc_r > \newsij{r}{j}$, which yields a contradiction of \rf{ass2}. Therefore $\val{\newSig}<\val{\Sigma}$.

Finally we remove the $y_\indxk$ column and its matched row $g_\indxk$. The resulting signature matrix still has $\val\newSig$, since $(n+\indxk,n+\indxk)\in\newT$ and  $\newsig_{n+\indxk,n+\indxk}=0$.\qed
\end{proof}

From the steps of applying the ES method, we can recover the original DAE by reverting the expression substitutions and removing the introduced variables $y_j$ and equations $g_j$. Similar to the LC method, the ES method also guarantees that, provided $v_\indxk\neq0$ for all $t$ in the interval of interest, the original DAE and the resulting one have (at least locally) the same solution (if there is one); this is shown in \cite[\S 6.3]{tgn2015aCAS}.  
It is again desirable to choose a column index $\indxk\in\nzsetES$, such that the $v_l$ is a (nonzero) constant. With this choice, the equivalence of the original DAE and the resulting one is {\em always} guaranteed. We hence derive a set $\newJes$, a subset of $\nzsetES$ that contains these $\indxk$'s for which $\indxk\in\nzsetES$  and $v_\indxk$ is constant.

\medskip

We summarize the steps of the ES method.
\begin{enumerate}[1)]
\item Obtain a symbolic form of $\Jac$.
\item Compute a $\vecv\in\ker\Jac$.
\item Derive $\nzsetES$,  $s$, $\eqsetIES$, $\overC$ as defined in \rf{ESdef}.
\item Check condition \rf{EScond}. If it is not satisfied, then set $\nzsetES\as\emptyset$ to mean that the ES method is not applicable; otherwise proceed to the next step.
\item $\newJes \as \setbg{ \indxk\in\nzsetES \mid \text{$v_\indxk$ is constant}}$. If $\newJes\neq\emptyset$, then choose an $\indxk\in\newJes$; otherwise an $\indxk \in\nzsetES$.
\item For each $j\in\Jnol$, introduce $y_j$, as defined in \rf{gjgk}, and append the corresponding equation $g_j$, as defined in \rf{gj}.
\item Replace each $x_j^{(d_j-c_i)}$ in $f_i$ by $\big( y_j+(\fracin{v_j}{v_\indxk}) \cdot x_\indxk^{(d_\indxk-\ESconst)} \big)^{(\ESconst-c_i)}$, for all $i\in\eqsetIES$ and all $j\in\Jnol$.
\item (Optional) For consistence, rename variables $y_j$, $j\in\Jnol$, to $x_{n+1},\ldots,x_{n+s-1}$, and rename equations $g_j$, $j\in\Jnol$, to $f_{n+1},\ldots,f_{n+s-1}$.
\end{enumerate}

The sets $\nzsetES$ and $\newJes$ are used to decide a desirable conversion method; see below.

\subsection{Which method to choose?}

We present our rationale for choosing a conversion method in \TBref{LCES} and base our choice on the following observations. For some failure cases, our experiments find that usually one of the LC condition \rf{LCcond} and the ES condition \rf{EScond} is satisfied, while the other is not, so we can apply one conversion method only. For other cases where both methods are applicable, we consider as priority the equivalence between the original DAE and the resulting one. As discussed in \SSCref{LC} and \SSCref{ES}, we wish to choose a nonzero constant $u_l$ [resp. $v_l$] for the LC [resp. ES] method. Our experience suggests that such a constant frequently exists for one of the methods. If both methods guarantee equivalence or neither of them does, then we choose the LC method, as it replaces only one existing equation and maintains the problem size. 

We summarize in \TBref{LCES} the above logic of finding the desirable conversion in the sense of equivalence. For instance, suppose the LC method finds $\newL=\emptyset$ and $\nzset\neq\emptyset$ while the ES method finds $\newJes\neq\emptyset$. Then either method can provide some conversion for reducing $\val\Sigma$. However, the LC method does not guarantee the equivalence while the ES method does, so we choose the ES method with a column index $\indxk\in\newJes$.
\begin{table}[!htb]
{\normalsize
\centering
\begin{tabular}{rc|ccc}
\mc{2}{c|}{\multirow{2}{*}{
}}& \mc{3}{c}{ES method}\\
& & $\newJes\neq\emptyset$ & $\newJes=\emptyset$ and $\nzsetES\neq\emptyset$  & $\nzsetES=\emptyset$ \\ \hline &&&\\[-2ex] 
\multirow{3}{*}{LC method} 
& $\newL\neq\emptyset$ & LC & LC & LC \\
& $\newL=\emptyset$ and $\nzset\neq\emptyset$ & ES & LC & LC \\
& $\nzset=\emptyset$ & ES & ES & --\\
\end{tabular}\caption{Rationale for choosing a conversion method. 
}\label{tb:LCES}
}
\end{table}

\section{More examples.}\label{sc:exam}
We show how to iterate the LC method on a linear constant coefficient DAE in \SSCref{coupledDAE}, and illustrate in \SSCref{pendES} the ES method with a modified pendulum problem by a linear transformation of the state variables.

\subsection{A linear constant coefficient DAE.}\label{ssc:coupledDAE}

Consider a linear constant coefficient DAE \cite[Example 10]{scholz2013combined}\footnote{
We consider it with parameters $\beta=\epsilon=1$, $\alpha_1=\alpha_2=\delta=1$, and $\gamma=-1$. Superscripts are used as indices there, while we use subscripts instead. The equations $g_1,g_2$ are renamed $f_3, f_4$, and the variables $y_1,y_2$ are renamed $x_3, x_4$.}  by Scholz and Steinbrecher
\begin{equation*}
\begin{alignedat}{3}
0=f_1 &= -&&x'_1 + x_3 &&- b_1\\
0=f_2 &= -&&x'_2 + x_4 &&- b_2
\end{alignedat}
\qquad\qquad
\begin{alignedat}{3}
0=f_3 &= &&x_2 + x_3 + x_4 &&- b_3\\
0=f_4 &= -&&x_1+ x_3 + x_4 &&- b_4\;.
\end{alignedat}
\end{equation*}

\begin{equation*}
{
\Siter{0} = 
\begin{blockarray}{r*4{@{\hspace{2mm}}c}@{\hspace{10pt}}ll}
 & \phanbull x_{1} & x_{2} & x_{3} & x_{4} &\s{c_i} \\
\begin{block}{r @{\hspace{2mm}}[*4{@{\hspace{2mm}}c}]@{\hspace{10pt}}ll}
 f_{1}&\phanbull1^\bullet&&0&&\s0\\ 
 f_{2}&&\phanbull1^\bullet&&0&\s0\\ 
 f_{3}&&\OK{0}&\phanbull0^\bullet&0&\s0\\ 
 f_{4}&\OK{0}&&0&\phanbull0^\bullet&\s0\\ 
\end{block}
\s{d_j} &\s1&\s1&\s0&\s0 \\[-4ex]
 \end{blockarray}
 \qquad
\Jiter{0} = \begin{blockarray}{r@{\hspace{2mm}}*4{@{\hspace{2mm}}c}}
 & x'_{1} & x'_{2} & x_{3} & x_{4} \\
\begin{block}{r@{\hspace{2mm}}[*4{@{\hspace{2mm}}c}]}
 f_{1}&-1&&1&\\ 
 f_{2}&&-1&&1\\ 
 f_{3}&&&1&1\\ 
 f_{4}&&&1&1\\ 
\end{block}
&\\[-4ex]
\end{blockarray}}
\end{equation*}
In each $f_i$ is a forcing function $b_i(t)$, $i=\rnge{1}{4}$.
We use a superscript in $\Siter{0}$ and $\Jiter{0}$ to mean an iteration number, not a power.
Since $\det(\Jiter{0})\equiv0$, the SA fails. 

\renewcommand\uiter[1]{\vecu}
\renewcommand\Kiter[1]{\nzset}
\renewcommand\kiter[1]{\indxk}
\renewcommand\thiter[1]{\LCconst}
\renewcommand\Iiter[1]{\eqsetI}

Choose  $\uiter{0}=[0,0,-1,1]^T\in\coker{\Jiter{0}}$.  Using \rf{LCdef} gives
\[
\Iiter{0}=\setbg{3,4},\quad \thiter{0}=0,\quad\text{and}\quad \Kiter{0}=\setbg{3,4}\;. 
\]
Since $\vecu$ is a constant vector, the condition \rf{LCcond} is surely satisfied, as $\hoder{x_j}{\vecu}=-\infty$ for all $j$. Choosing $\kiter{0}=3\in\Kiter{0}$ and replacing $f_3$ by $\newf_3$ results in $0=(f_1,f_2,\newf_3,f_4)$,~where
\begin{align*}
\newf_3 = \sum_{i\in \eqsetI} u_i f_i^{\left(c_i-\LCconst\right)}
=  -f_3+f_4 = -x_1 -x_2 + b_3 - b_4.\\[-4ex]
\end{align*}
\begin{align*}
\Siter{1}=
\begin{blockarray}{rc*3{@{\hspace{2mm}}c}ll}
 & x_{1}\phanbull & x_{2} & x_{3} & x_{4} &\s{c_i} \\
\begin{block}{r [c*3{@{\hspace{2mm}}c}]ll}
f_1 & 1^\bullet  &  &0\phanbull  &  &   \s0 \\ 
f_2 &  &1\phanbull  &  &0^\bullet &   \s0  \\ 
\newf_3 & 0\phanbull  &0^\bullet &  && \s1  \\
f_4 & \lgo\phanbull &  & 0^\bullet &0\phanbull & \s0  \\
\end{block}
 \s{d_j}& \s1\phanbull &\s1\phanbull &\s0\phanbull &\s0\phanbull \\[-4ex]
 \end{blockarray}
\qquad\Jiter{1} = \begin{blockarray}{rc*3{@{\hspace{2mm}}c}@{\hspace{10pt}}ll}
 & x'_{1} & x'_{2} & x_{3} & x_{4} \\
\begin{block}{r @{\hspace{2mm}}[c*3{@{\hspace{2mm}}c}]@{\hspace{10pt}}ll}
f_1 & -1  &  &1  &  &    \\ 
f_2 &   & -1  &  &1 &    \\ 
\newf'_3 &-1   &-1 &  &  &   \\
f_4 &  &  & 1 & 1 &   \\
\end{block}
&\\[-4ex]
\end{blockarray}
\end{align*}

The SA fails still since $\det(\Jiter{1})\equiv0$. We then apply the LC method again by choosing $\uiter{1}=[-1,-1,1,1]^T\in \coker{\Jiter{1}}$. This gives
\[
\Iiter{1}=\setbg{1,2,3,4},\quad \thiter{1}=0,\quad\text{and}\quad \Kiter{1}=\setbg{1,2,4}\;.
\]
Choosing $\kiter{1}=1\in\Kiter{1}$ and replacing $f_1$ by
\begin{align*}
\newf_1  =  -f_1-f_2+\newf'_3 +f_4
=-x_1+b_1+b_2+b'_3-b'_4-b_4
\end{align*}
results in $0=(\newf_1,f_2,\newf_3,f_4)$.
\begin{align*}
\Siter{2} =
\begin{blockarray}{rc*3{@{\hspace{2mm}}c}ll}
& x_1 &  x_2 &  x_3 & x_4 & \s{\newc_i} \\
\begin{block}{r [c*3{@{\hspace{2mm}}c}]ll}
\newf_1 & 0^\bullet &  &  && \s1\\
f_2 &   &1\phanbull  &  &0^\bullet &   \s0  \\
\newf_3 &  0\phanbull & 0^\bullet &  && \s1 \\
f_4 & \lgo\phanbull  &  &0^\bullet  & 0\phanbull &   \s0\\ 
\end{block}
\s{\newd_j}&\s1\phanbull &\s1\phanbull & \s0\phanbull &\s0\phanbull  \\[-4ex]
\end{blockarray}
\qquad
\Jiter{2} = \begin{blockarray}{rc*3{@{\hspace{2mm}}c}@{\hspace{10pt}}ll}
& x'_1 &  x'_2 &  x_3 & x_4 &  \\
\begin{block}{r [c*3{@{\hspace{2mm}}c}]@{\hspace{10pt}}ll}
\newf'_1 &-1  &  &  &    \\
f_2 &  & -1  &  & 1    \\ 
\newf'_3 & -1  & -1 &  &   \\
f_4 &  & & 1 & 1     \\ 
\end{block}
&\\[-4ex]
\end{blockarray}
\end{align*}
The SA succeeds since $\detJac{\Jiter{2}}{1}$. Note $\valSig{\Siter{2}}{0} < \valSig{\Siter{1}}{1} < \valSig{\Siter{0}}{2}$.

\subsection{Modified pendulum by change of variables.}\label{ssc:pendES}
For the pendulum DAE \rf{pend}, we perform a linear transformation on  $x,y,\lam$:
\begin{align*}
\left[
\begin{array}{c}
x\\y\\ \lam
\end{array}
\right] = \left[
\begin{array}{ccc}
1 & 1 & 0\\ 0 & 1 & 1\\ 1 & 0 & 1
\end{array}
\right] \left[
\begin{array}{c}
x_1\\ x_2\\ x_3
\end{array}
\right]\;.
\end{align*} 
The resulting problem is
\begin{equation}\label{eq:pendabc}
\begin{aligned}
0 = f_1 &= (x_1+x_2)''+(x_1+x_2)(x_3+x_1)\\
0 = f_2 &= (x_2+x_3)''+(x_2+x_3)(x_3+x_1) -\pendG\\
0 = f_3 &= (x_1+x_2)^2+(x_2+x_3)^2-\pendL^2\;.\\[-2ex]
\end{aligned}
\end{equation}
\begin{align*}
\Sigma=
\begin{blockarray}{rc@{\hspace{2mm}}c@{\hspace{2mm}}c ll}
& x_1 &  x_2 &  x_3 &  \s{c_i} \\
\begin{block}{r @{\hspace{10pt}}[c@{\hspace{2mm}}c@{\hspace{2mm}}c]ll}
f_1 & 2  & 2 & \OK{0}   &   \s0 \\
f_2 & \OK{0} &2  &2 &   \s0  \\ 
f_3 & 0  &0 & 0 & \s2  \\
\end{block}
 \s{d_j}& \s2 &\s2 &\s2 \\[-4ex]
 \end{blockarray}
\qquad\Jac = \begin{blockarray}{rc@{\hspace{2mm}}c@{\hspace{2mm}}ccc}
& x''_1 &  x''_2 & x''_3 \\
\begin{block}{r @{\hspace{10pt}}[c@{\hspace{2mm}}c@{\hspace{2mm}}c]cc}
f_1 & 1    &1    &      \\
f_2 &  &1    &1      \\
f''_3 & 2\alpha  &2(\alpha+\beta)    &2\beta    \\  
\end{block}
\end{blockarray}
\end{align*}
Here $\alpha=x_1+x_2$ and $\beta=x_2+x_3$. That $\det(\Jac)\equiv0$ is expected.

We first attempt the LC method and compute $\vecu=[2\alpha,2\beta,-1]^T\in\coker\Jac$. Using \rf{LCdef} it finds
\[
\eqsetI=\setbg{\,i \mid\, u_i\nequiv 0\,}=\setbg{\rnge{1}{3}},\quad 
\LCconst=\min_{i\in \eqsetI}c_i=0,\quad
\nzset=\setbg{\, i \in \eqsetI \mid\, c_i=\underC} =\setbg{1,2}\;.
\]
For all $\indxk\in\nzset$, $u_\indxk$ is not a constant, so $\nzset\neq\emptyset$ and $\newL=\emptyset$.  Then we try the ES method to seek a conversion that guarantees equivalence.

We show below how the ES method reveals the linear transformation of the states without having knowledge about the equations.

Compute $\vecv=[1,-1,1]^T\in\ker\Jac$ and, using \rf{ESdef}, find
\begin{align*}
\nzsetES &=\setbg{1,2,3},\quad s=|\nzsetES|=3,\quad \eqsetIES=\setbg{1,2,3},\quad \text{and} \quad \ESconst = 2\;.
\end{align*}
Since
\[
\hoder{x_j}{\vecv}=-\infty \qaq d_j-\ESconst=0 \quad\text{for all $j$}\;,
\]
the ES condition \rf{EScond} is satisfied, and we have $\newJes=\nzsetES=\setbg{1,2,3}$. Below shows the case $\indxk=1\in\newJes$ for example.

As $\nzsetES\setminus \setbg{\indxk} =\setbg{2,3}$, we introduce new variables $y_2$ and $y_3$ for $x_2$ and $x_3$, respectively. Then \rf{gjgk} becomes
\begin{align*}
y_2 &= x_2^{(d_2-\ESconst)}  - (\fracin{v_2}{v_1})\cdot x_1^{(d_1-\ESconst)} 
= x_2+x_1  \\
\text{and}\quad
y_3 &= x_3^{(d_3-\ESconst)} - (\fracin{v_3}{v_1})\cdot x_1^{(d_1-\ESconst)} 
= x_3-x_1\;.
\end{align*}
The corresponding two equations are
\begin{equation*}
0 = g_2 = -y_2 + x_2+x_1 \qaq 0 = g_3 = -y_3 + x_3 - x_1\;.
\end{equation*}

We first  write {\em explicitly} the derivatives $x''_1,\ x''_2$ and $x''_3$ in $f_1$ and $f_2$:
\begin{align*}
0 = f_1 = x''_1+x''_2 + (x_1+x_2)(x_3+x_1),\quad
0 = f_2 = x''_2+x''_3 + (x_2+x_3)(x_3+x_1) -\pendG\;.
\end{align*}
Then we perform expression substitutions as described in the table.
\[
  \begin{array}{l@{\hspace{5mm}}l@{\hspace{5mm}}l}
 \text{substitute} & \text{for} &\text{in}  \\ \hline
y''_2-x''_1 & x''_2 & f_1,\; f_2\\
y''_3+x''_1 & x''_3 & f_2\\
y_2-x_1 & x_2 & f_3\\
y_3+x_1 & x_3 & f_3
  \end{array}
\]

One may want to make the variable names consistent and do the same for the equation names. By Step 8 of the ES method, variables $y_2,y_3$ are renamed $x_4,x_5$, while equations $g_2,g_3$ are renamed $\newf_4,\newf_5$. The resulting DAE is
\begin{equation}\label{eq:pendabcES}
\begin{alignedat}{3}
0 = \newf_1 &= x_4''+x_4(2x_1+x_5)\\
0 = \newf_2 &= (x_4+x_5)''+(x_4+x_5)(2x_1+x_5) -\pendG\\
0 = \newf_3 &= x_4^2+(x_4+x_5)^2-\pendL^2
\end{alignedat}
\qquad
\begin{alignedat}{3}
0 = \newf_4 &= -x_4+x_2+x_1 \\
0 = \newf_5 &= -x_5+x_3+x_1\;. \\
&
\end{alignedat}
\end{equation}
\begin{align*}
\newSig = \begin{blockarray}{r@{\hspace{10pt}}c*4{@{\hspace{2mm}}c}ll}
 & x_{1}\phanbull & x_{2} & x_{3} & x_{4} & x_{5} & \s{c_i} \\
\begin{block}{r @{\hspace{10pt}}[c*4{@{\hspace{2mm}}c}]ll}
 \newf_{1}&0\phanbull &&&2^\bullet&\OK{0}\phanbull&\s0\\ 
 \newf_{2}&0^\bullet&&&2\phanbull &2\phanbull &\s0\\ 
 \newf_{3}&&&&0\phanbull &0^\bullet&\s2\\ 
 \newf_{4}&0\phanbull &0^\bullet&&\OK{0}\phanbull&&\s0\\ 
 \newf_{5}&0\phanbull &&0^\bullet&&\OK{0}\phanbull&\s0\\ 
\end{block}
\s{d_j} &\s0\phanbull &\s0\phanbull &\s0\phanbull &\s2\phanbull &\s2\phanbull \\[-4ex]
 \end{blockarray}
 \qquad
\newJ = \begin{blockarray}{r*5{c@{\hspace{2mm}}}cc}
 & x_{1} & x_{2} & x_{3} & x''_{4} & x''_{5} \\
\begin{block}{r @{\hspace{10pt}}[*5{c@{\hspace{2mm}}}]cc}
 \newf_{1}&2x_4&&&1&\\ 
 \newf_{2}&2\mu &&&1&1\\ 
 \newf''_{3}&&&&2(x_4+\mu) &2\mu \\ 
 \newf_{4}&-1&-1&&&\\ 
 \newf_{5}&1&&-1&&\\ 
\end{block}
\end{blockarray}
\end{align*}
Here $\mu=x_4+x_5$. The SA succeeds on \rf{pendabcES}, since by $\newf_3=0$, we have
\[
\det(\newJ)=-4(2x_4^2+2x_4 x_5+x_5^2)=-4\pendL^2\neq 0\;.
\]

\section{Conclusions.}\label{sc:conclu}

We proposed two conversion methods for improving the \Sigmeth. They convert a DAE with finite $\val{\Sigma}$ and an identically (but not structurally) singular System Jacobian to a DAE that may have a nonsingular System Jacobian. A conversion guarantees that both DAEs have equivalent solutions (if any).
The conditions for applying these methods can be checked automatically, and the main result of a conversion is $\val{\newSigma}<\val{\Sigma}$, where $\newSig$ is the signature matrix of the resulting DAE.

An implementation of these methods requires the following steps.
\begin{enumerate}[1)]
\item  Compute a symbolic form of a System Jacobian $\Jac$.
\item  Find a vector in $\coker\Jac$ [respectively $\ker\Jac$].
\item  Check the LC condition \rf{LCcond} [respectively ES conditions \rf{EScond}].
\item  Generate the equations for the resulting DAE.
\end{enumerate}

In general, the computational cost of a conversion depends on the size of the DAE, sparsity, and intricacy of the equations. Determining the cost in advance is undecidable due to the results in \cite{richardson1968}. For example, fixing $\~M\~f=0$ can be arbitrarily difficult, where $\~f=0$ is a solvable DAE and $\~M$ is a generically nonsingular dense $n\times n$ matrix that can contain any derivatives of the $x_j$'s, typically lower than the $d_j$th. So far, all the fixes we have found for those failure cases are not difficult to compute.

In \cite{tgn2015d}, a continuation of this paper, we shall show how to combine the conversion methods with block triangularization of a DAE. For DAEs whose $\Jac$ can be permuted into a block-triangular form \cite{NedialkovPryce2012a, pryce2014btf}, we can locate the diagonal blocks that are singular and then apply a conversion to each such block, instead of working on the whole DAE. This approach improves the efficiency of finding a useful conversion for reducing $\val\Sigma$. Using our block conversion methods, we shall show the remedies for the Campbell-Griepentrog robot arm in \cite{campbell1995solvability}, and the transistor amplifier and the ring modulator in \cite{TestSetIVP}.

\appendix
\section{Preliminary results and proof of \protect\LEref{ESnewSig}.}\label{ap:proof}

{

Let the notation be as at the start of \SSCref{ES}. We give two preliminary lemmas prior to the main proof of \LEref{ESnewSig}.

\begin{lemma}\label{le:prelim1}
Let $r\in\Jnol$ and 
$\omega_1=y_\jtwo+(\fracin{v_\jtwo}{v_\indxk})\cdot x_\indxk^{(d_\indxk-\ESconst)}$.
Then
\begin{equation}\label{eq:prelim1}
\hoder{x_\jone}{\omega_1}=
\casemod{ll}{
<   d_j-\ESconst& \quad\text{if $j\in \Jnol$}\\[1ex]
\le d_j-\ESconst & \quad\text{otherwise}\;.
}   
\end{equation}
\end{lemma}
\begin{proof}
Consider the case $\jone=\indxk\in\nzsetES$. Obviously $\hoder{x_\indxk}{\omega_1}=d_\indxk-\ESconst$.

Now consider the case $\jone\neq\indxk$. Since $x_\jone$ can occur only in $v_\jtwo$ and $v_\indxk$ in $\omega_1$, we have $\hoder{x_\jone}{\omega_1} \le \hoder{x_\jone}{\vecv}$. It follows from \rf{EScond} and the case $\jone=\indxk$ that \rf{prelim1} holds.\qed

\end{proof}

\begin{lemma}\label{co:prelim2}
Let $r\in\Jnol$, $i\in\eqsetIES$, and 
\begin{equation}\label{eq:omega2}
\omega_2 
= \omega_1^{(\ESconst-c_i)}
= \left(y_\jtwo+\frac{v_\jtwo}{v_\indxk}\cdot x_\indxk^{(d_\indxk-\ESconst)}\right)^{(\ESconst-c_i)}\;.
\end{equation}
 Then
\begin{equation}\label{eq:prelim2}
\hoder{x_\jone}{\omega_2}=
\casemod{ll}{
<   d_j-c_i& \quad\text{if $j\in \Jnol$}\\[1ex]
\le d_j-c_i & \quad\text{otherwise}\;.
}   
\end{equation}
\end{lemma}

\begin{proof}
Since $\ESconst=\max_{i\in \eqsetIES} c_i$, we have $\ESconst-c_i\ge 0$ for all $i\in \eqsetIES$. From \rf{omega2}, connecting $\hoder{x_j}{\omega_2}=\hoder{x_j}{\omega_1}+(\ESconst-c_i)$ to \rf{prelim1} immediately yields \rf{prelim2}.\qed
\end{proof}

Using the two assumptions before \LEref{ESnewSig}, we prove it below. 

\begin{proof}
Write $\newSig$ in \FGref{ESnewSig} into the following $2\times 3$ block form:
\begin{align}
\newSig = \left[
\renewcommand{\arraystretch}{1.5}
\begin{array}{c|c|c}
\newSigBlk{1}{1} & \,\newSigBlk{1}{2} & \,\newSigBlk{1}{3} \\ \cline{1-3}
\newSigBlk{2}{1} & \,\newSigBlk{2}{2} & \,\newSigBlk{2}{3}
\end{array}
\right]\;.
\end{align}
We aim to verify below the relations between $\newsij{i}{j}$ and $\newd_j-\newc_i$ in each block. 

\begin{enumerate}[(1)]
\item $\newSigBlk{1}{1}$. Consider $\jone,\jtwo\in\Jnol$. By \rf{ESsubs}, we substitute $\omega_2$ in \rf{omega2} for every $x_\jtwo^{(d_\jtwo-c_i)}$ in $f_i$ for all $i=\rnge{1}{n}$.
By \rf{prelim2}, $\hoder{x_\jone}{ \omega_2} <  d_\jone - c_i$ for all $i\in\eqsetIES$.
So these expression substitutions do not introduce $x_\jtwo^{(d_\jtwo-c_i)}$ in $\newf_i$, where $\jtwo\in\Jnol$. Given $M$ in \rf{ESdef}, we have $d_j-c_i>\sij{i}{j}$ for all $i\notin\eqsetIES$ and $j\in\nzsetES$. Hence
\begin{align}\label{eq:Sig11}
\hoder{x_\jone}{\newf_i} < d_\jone-c_i
\qquad \text{for $\jone\in\Jnol, \ i=\rnge{1}{n}$}\;.  
\end{align}

What remains to show is the case $j=\indxk$.
From \rf{gjgk},
\begin{equation*}
x_r^{(d_r-\ESconst)}=y_r+ \frac{v_\jtwo}{v_\indxk} \cdot x_\indxk^{(d_\indxk-\ESconst)}\;.
\end{equation*}
Taking the partial derivatives of both sides with respect to $x_\indxk^{(d_\indxk-\ESconst)}$ and applying Griewank's Lemma \rf{griewank} with $w=x_{\jtwo}^{(d_{\jtwo}-\ESconst)}$ and $q=\ESconst-c_i\ge 0$ for all $i\in \eqsetIES$, we have
\begin{align}
\frac{v_{\jtwo}}{v_\indxk}
= \pp{x_{\jtwo}^{(d_{\jtwo}-\ESconst)}}{x_\indxk^{(d_\indxk-\ESconst)}}
= \pp{x_{\jtwo}^{(d_{\jtwo}-\ESconst+\ESconst-c_i) }}{x_\indxk^{(d_\indxk-\ESconst+\ESconst-c_i)}}
= \pp{x_{\jtwo}^{(d_{\jtwo}-c_i)}}{x_\indxk^{(d_\indxk-c_i)}}\;.
\label{eq:ujuk}
\end{align}
Then
\begin{equation*}
\begin{alignedat}{3}
\pp{\newf_i}{x_\indxk^{(d_\indxk-c_i)}} 
&=  \pp{f_i}{x_\indxk^{(d_\indxk-c_i)}} + \sum_{r\in \nzsetES\setminus\{\indxk\}}
\pp{f_i}{x_r^{(d_r-c_i)}} \cdot 
\pp{x_r^{(d_r-c_i)}}{x_\indxk^{(d_\indxk-c_i)}}
&\qquad&\text{by the chain rule}  \nonumber\\
&= \Jij{i}{\indxk} + \sum_{r\in \nzsetES\setminus\{\indxk\}}
\Jij{i}{r} \cdot \frac{v_r}{v_\indxk} 
\nonumber &&\text{by \rf{ujuk}} \nonumber\\
&= \frac{1}{v_\indxk} \sum_{r\in \nzsetES}  \Jij{i}{r} v_r= \frac{1}{v_\indxk}(\~J \vecv)_i = 0 &&\text{by $\~J \vecv=\~0$}\;.
\end{alignedat}
\end{equation*}
This gives $\hoder{x_\indxk}{\newf_i} < d_\indxk-c_i$ for all $i=\rnge{1}{n}$.
Together with \rf{Sig11} we have proved the ``$<$'' part in $\newSig_{11}$.

\item $\newSigBlk{1}{2}$. The substitutions do not affect $x_j$, for all $j\notin\nzset$. By \rf{prelim2}, such an $x_j$ occurs in every $ \omega_2$ of order $\le d_j-c_i$, where $i\in\eqsetIES$. Hence also
\[
\hoder{x_j}{\newf_i} \le d_j-c_i \quad\text{for all $i=\rnge{1}{n}$ and $j\notin\nzset$}\;.
\]

\item $\newSigBlk{1}{3}$. Consider $\jtwo\in\Jnol$. For an $i\in\eqsetIES$, $y_\jtwo$ occurs of order $\ESconst-c_i$ in $\omega_2$ in \rf{omega2}. For all $i=\oneton$, if a substitution occurs for an $x_\jtwo^{(d_\jtwo-c_i)}$ in $f_i$, then $\hoder{y_\jtwo}{\newf_i}=\ESconst-c_i$; otherwise $\hoder{y_\jtwo}{\newf_i}=-\infty$. In either case $\hoder{y_\jtwo}{\newf_i}\le \ESconst-c_i$.

\item $\newSigBlk{2}{1}$. Equalities hold on the diagonal and in the $\ssth{\indxk}$ column, as $y_\jtwo^{(d_\jtwo-\ESconst)}$ and $y_\indxk^{(d_\indxk-\ESconst)}$ occur in $g_\indxk$, where $\jtwo\in\nzsetES$. What remains to show is the ``$<$'' part. Assume that $\jone\,,\jtwo\,,\indxk\in\nzsetES$ are distinct. Then by~\rf{gjgk}~and~\rf{EScond},
\begin{align}\label{eq:Sig21}
\hoder{x_{\jone}}{g_{\jtwo}} 
&=\hoder{x_{\jone}}{ y_{\jtwo}-x_{\jtwo}^{(d_{\jtwo}-\ESconst)}
+ \frac{v_{\jtwo}}{v_\indxk} \cdot x_\indxk^{(d_\indxk-\ESconst)}} \le \hoder{x_{\jone}}{\vecv} < d_\jone - \ESconst\;.
\end{align}

\item $\newSigBlk{2}{2}$. Assume again that $\jone\,,\jtwo\,,\indxk$ are distinct, where $\jtwo\in\nzsetES$ and $j=\rnge{s+1}{n}$. Then replacing the ``$<$'' in \rf{Sig21} by ``$\le$'' proves the ``$\le$'' part in $\newSigBlk{2}{2}$.

\item $\newSigBlk{2}{3}$. Consider $\jtwo, \jone\in\nzsetES$.
By  $0=g_\indxk = -y_\indxk + x_\indxk^{(d_\indxk-\ESconst)}$ and \rf{gjgk}, $y_\jone$ occurs in $g_\jtwo$ only if $\jone=\jtwo$, and $\hoder{y_\jone}{g_\jone}=0$. Hence, on the diagonal lie zeros, and everywhere else is filled with $-\infty$.
\end{enumerate}

Also worth noting is that in the $y_\indxk$ column is only one finite entry $\sigma_{n+\indxk,n+\indxk}=0$, and that in the $g_\indxk$ row are only two finite entries $\sigma_{n+\indxk,n+\indxk}=0$ and $\sigma_{n+\indxk,\indxk}=d_\indxk-\ESconst$.

Recalling \rf{newcdES} for the formulas of $\newc_i$ and $\newd_j$ of $\newSig$, we can summarize that the above items (1)--(6) verify the relations between $\newsij{i}{j}$ and $\newd_j-\newc_i$ in $\newSigma$ for all $i,j=\rnge{1}{n+s}$; see \FGref{ESnewSig}. \qed

\end{proof} 
}

\begin{acknowledgements}
The authors acknowledge with thanks the financial support for this research: GT is supported in part by the McMaster Centre for Software Certification through the Ontario Research Fund, Canada, NSN is supported in part by the Natural Sciences and Engineering Research Council of Canada, and JDP is supported in part by the Leverhulme Trust, the UK. 
\end{acknowledgements}

\end{document}